\documentclass[sigconf,nonacm]{acmart}
\usepackage{hyperref}
\usepackage{xspace}
\usepackage{xparse}
\usepackage{comment}
\usepackage{graphicx}
\usepackage{amsfonts}
\usepackage{latexsym}
\usepackage{amsmath}
\usepackage{amsthm}
\usepackage{nopageno}
\usepackage{mdwlist}
\usepackage[noend]{algpseudocode}
\usepackage{algorithm}
\usepackage{booktabs}
\usepackage{framed}
\usepackage{listings}
\usepackage{color}
\usepackage{stmaryrd}
\usepackage{tikz-cd}
\usepackage{mathrsfs}
\usepackage{thmtools}
\usepackage{enumitem}
\usepackage{balance}
\newtheorem{definition}{Definition}

\newtheorem{lemma}{Lemma}
\newtheorem{theorem}{Theorem}

\newcommand{\eat}[1]{}
\newcommand{\pages}[1]{} 

\newcommand{\eg}{\emph{e.g.}\xspace}

\newcommand{\ie}{\emph{i.e.}\xspace}

\newcommand{\etal}{{\em et~al.}\xspace}
\newcommand{\passiveObfs}{FEP-CPFA\xspace}
\newcommand{\bothObfs}{FEP-$x$}
\NewDocumentCommand{\activeObfs}{o}{\IfNoValueTF{#1}{FEP-CCFA\xspace}{FEP-CCFA-\tf{#1}\xspace}
}
\newcommand{\shaping}{Traffic Shaping\xspace}

\newcommand{\tf}[1]{\textsc{#1}}
\newcommand{\tok}[2][]{\texttt{#2}_{#1}}
\newcommand{\emptySt}{\epsilon}
\newcommand{\errDec}{\bot_{\tf{Dec}}}
\newcommand{\rgets}{\gets_\$}

\newcommand{\lenlen}{\ell_{\tok{len}}}
\newcommand{\cat}{\bigparallel}

\newcommand{\indcca}{IND-CCA}
\newcommand{\indcpa}{IND-CPA}
\newcommand{\indcxa}{IND-$x$}
\newcommand{\intctxt}{INT-CTXT}
\newcommand{\indlor}{LoR-IND}
\newcommand{\indccarecv}{Recv-IND}
\newcommand{\intctxtsend}{Send-INT}
\newcommand{\intctxtrecv}{Recv-INT}
\newcommand{\bothfepdgram}{FEP-$x$}
\newcommand{\fepcpa}{FEP-CPA}
\newcommand{\fepcca}{FEP-CCA}
\newcommand{\fepdgramsend}{Send-DG}
\newcommand{\fepdgramrecv}{Recv-DG}
\newcommand{\nullLength}{\ell_{\textrm{null}}}
\newcommand{\inLength}{\ell_{\textrm{in}}}
\newcommand{\outLength}{\ell_{\textrm{out}}}
\newcommand{\nonceSize}{\ell_{\textrm{Nonce}}}
\newcommand{\tagSize}{\ell_{\textrm{Tag}}}
\newcommand{\aeadOverhead}{\ell_{\textrm{Overhead}}}
\newcommand{\messageSpace}{\mathcal{M}_{\inLength}}
\newcommand{\emptyList}{[]}
\newcommand{\getsAppend}[2]{\gets \textsc{Append}(#1, #2)}
\newcommand{\nullsym}{\top}

\setcopyright{usgov}
\copyrightyear{2024}
\acmYear{2024}
\acmDOI{3658644.3690198}
\acmConference[CCS '24]{Proceedings of the 2024 ACM SIGSAC Conference on Computer and Communications Security}{October 14--18, 2024}{Salt Lake City, UT, USA}
\acmBooktitle{Proceedings of the 2024 ACM SIGSAC Conference on Computer and Communications Security (CCS '24), October 14--18, 2024, Salt Lake City, UT, USA}
\acmDOI{10.1145/3658644.3690198}
\acmISBN{979-8-4007-0636-3/24/10}
\newif\iffull
\fulltrue

\begin{document}

\title[Bytes to Schlep? Use a FEP: Hiding Protocol Metadata with Fully Encrypted Protocols]{Bytes to Schlep? Use a FEP:\\Hiding Protocol Metadata with Fully Encrypted Protocols}

\author{Ellis Fenske}
\orcid{0000-0003-2955-9521}
\affiliation{
	\institution{U.S. Naval Academy}
        \city{Annapolis}
        \state{Maryland}
	\country{U.S.A.}
}
\email{fenske@usna.edu}

\author{Aaron Johnson}
\orcid{0000-0002-2057-1690}
\affiliation{
	\institution{U.S. Naval Research Laboratory}
        \city{Washington}
        \state{D.C.}
	\country{U.S.A.}
}
\email{aaron.m.johnson213.civ@us.navy.mil}
\renewcommand{\shortauthors}{Ellis Fenske and Aaron Johnson}


\begin{abstract}
Fully Encrypted Protocols (FEPs) have arisen in practice as a technique to avoid network
censorship. Such protocols are designed to produce messages that appear completely random. This
design hides communications metadata, such as version and length fields, and makes it difficult to
even determine what protocol is being used. Moreover, these protocols frequently support padding to
hide the length of protocol fields and the contained message. These techniques have relevance well
beyond censorship circumvention, as protecting protocol metadata has security and privacy benefits
for all Internet communications. The security of FEP designs depends on cryptographic assumptions,
but neither security definitions nor proofs exist for them. We provide novel security definitions
that capture the metadata-protection goals of FEPs. Our definitions are given in both the datastream
and datagram settings, which model the ubiquitous TCP and UDP interfaces available to protocol
designers. We prove relations among these new notions and existing security definitions. We further
present new FEP constructions and prove their security. Finally, we survey existing FEP candidates
and characterize the extent to which they satisfy FEP security. We identify novel ways in which these protocols are identifiable, including their responses to the introduction of data errors and the sizes of their smallest protocol messages.
\end{abstract}

\maketitle

\section{Introduction}
One approach to avoid network censorship is to use a Fully Encrypted Protocol (FEP). FEPs are
designed to hide communications metadata, such as the precise length of a plaintext message, the
encryption algorithms being used, and the exact protocol being run. A goal of FEPs is that each byte
appears uniformly random (a design sometimes referred to as ``obfuscated'',
``look-like-nothing'', ``randomized'', or ``unfingerprintable''~\cite{dixon2016network,wu2023great,wang2015seeing,perrin2018noise}).
Censorship circumvention tools using FEPs include the popular Shadowsocks~\cite{shadowsocks},
obfs4~\cite{obfs4}, and VMess~\cite{vmess} systems, each of which provides its own unique FEP.
The Noise protocol framework~\cite{perrin2018noise} requires ciphertexts be indistinguishable from random to be ``censorship-resistant'' (although it fails to specify all
protocol aspects needed to satisfy that goal). While
methods to identify and block FEPs exist~\cite{wang2015seeing,wails:ndss:24}, FEPs
nonetheless continue to be effective and popular~\cite{shadowsocks,v2ray,obfs4}.

FEPs are also helpful outside of a censorship setting. Hiding protocol metadata can improve security
(\eg, preventing attacks targeted at a specific protocol or implementation) and privacy (\eg,
preventing traffic analysis of message lengths). Indeed, standard encrypted transport protocols like
TLS and QUIC have moved over time towards encrypting more protocol metadata, as it is repeatedly
observed that seemingly innocuous metadata is unexpectedly sensitive. The Pseudorandom cTLS
extension~\cite{ctls-pseudorandom} is a FEP that has recently been proposed to the IETF, citing
security and privacy as its main motivations. In addition, the use of FEPs can prevent Internet
ossification because they provide little metadata for middleboxes to operate on, a strategy similar
to the use of random extensions in GREASE~\cite{benjamin2020rfc}. Thus, we consider FEPs to be a
natural endpoint of encrypted transport protocol development.

Despite these motivations, FEPs have primarily been developed by the open-source community.
Consequently, real-world protocols have been developed without mathematical security goals, even
though their security depends on cryptographic assumptions\footnote{The Pseudorandom cTLS
proposal~\cite{ctls-pseudorandom} cites the need for such formalization: ``TODO: More precise
security properties and security proof. The goal we're after hasn't been widely considered in the
literature so far, at least as far as we can tell.''}. Also, in practice, different developers make
a variety of design errors~\cite{fifield2023flaws}, as they lack a common set of FEP techniques
and pitfalls.

Therefore, we propose novel, precise security definitions for FEPs. We formalize the goal of
producing apparently random protocol messages with a passive adversary, and then we extend that to
the goal of appearing otherwise predictable (and thus not leaking information) with respect to an
active adversary. The predictability applies to the protocol's behavior when ciphertexts are
modified and to the protocol's use of \emph{channel closures}, an observable feature of
connection-based transport protocols such as QUIC and TCP. Research has shown that channel closures
can be used to identify specific FEPs~\cite{frolov2020detecting}, and real-world measurements
indicate that censors are already doing so~\cite{beznazwy2020china}.

In addition, most FEPs (and many other encrypted protocols) include message padding as a way to
prevent traffic analysis from making inferences on traffic contents based on ciphertext lengths. For
FEPs, an effective padding strategy is particularly powerful because each byte appears random, and
so the number of bytes is the main remaining source of communications metadata. However, existing
padding mechanisms are ad hoc and have no precisely stated goals. As a consequence, they fail to
provide full control over over ciphertext lengths, a problem that has been observed by FEP
developers themselves~\cite{fifieldTordev}, which serves as the basis for state-of-the-art FEP detection
using packet lengths~\cite{wails:ndss:24}, and has been exploited in practice by
censors~\cite{beznazwy2020china}. We thus propose a powerful property called \emph{\shaping}, which
requires that protocols be capable of producing ciphertexts of arbitrary lengths on command.

We prove relations among our novel FEP security definitions and existing notions of
confidentiality and integrity. Our analysis of existing FEPs indicates that none of them satisfies
our notions, and so we construct a protocol and prove that it does. These results are
given in the \emph{datastream} setting, which models a protocol that uses an underlying
reliable, byte-oriented transport such as TCP. We subsequently give definitions and related results
in the \emph{datagram} setting, which models the use of an unreliable transport such as UDP, and we
present a provably secure FEP in it as well.

Our theoretical results are inspired by observing persistent problems in real-world FEP 
deployments and our work provides concrete guidance for future FEP development. Our security 
definitions provide explicit security goals for FEPs to satisfy, and allow FEP maintainers
to prove their protocols secure. Our \shaping notions give FEP designers a length obfuscation
capability that is maximally flexible and practical to implement. In addition our definitions of 
secure close functions highlight the importance of carefully implementing connection closures,
and explicitly provides a narrow set of possible close behaviors which do not reveal metadata unnecessarily. 

Finally, we analyze a wide variety of existing FEPs under our proposed security definitions,
including the previously mentioned Shadowsocks~\cite{shadowsocks}, obfs4~\cite{obfs4}, and
VMess~\cite{vmess} protocols. We examine their source code and documentation, and we run experiments
to measure their error responses and message sizes. Our results indicate that nearly all of them do
produce outputs that are indistinguishable from randomness, satisfying our passive FEP security
definition. However, among the datastream protocols, their channel closures in response to errors
make them identifiable, and they consequently fail our active FEP security definition. Our
experiments further uncover integrity failures in V2Ray (the system providing the VMess
protocol~\cite{v2ray}). Moreover, all of the FEPs exhibit \emph{unique minimum message sizes}, which
we both predict via source-code analysis and verify experimentally. The catalog we produce of
channel closures and minimum message sizes provides a new method by which existing FEPs can be
individually identified and thus blocked. It also supports our proposed security definitions, which
normalize channel closures and preclude minimum message sizes.

The main contributions of our paper are thus as follows. First, we present security definitions for datastream FEPs, compare them to existing security definitions, and provide a provably secure datastream FEP construction. Second, we similarly provide and analyze datagram FEP definitions, and we provide a provably secure datagram FEP. Third, we evaluate many existing FEPs against our security definitions, uncovering novel identifying features in both error-induced channel closures and minimum message sizes.

This paper improves and expands on an early version of this work~\cite{focipaper} by including
a more developed notion of close functions, security definitions and a construction in the datagram setting, analyses of relations between the new security notions and existing ones, and an analysis of and experiments on existing FEPs. 
\iffull
  The appendices of this paper also contain proofs and other details.
\fi

\section{Background and Preliminaries} \label{sec:aead-scheme}
We present the notation, primitives, and concepts required to describe our
constructions, definitions, and counterexamples.

We use $\emptySt$ to denote the empty byte string and $s\Vert t$ to denote the concatenation of two
byte strings $s$ and $t$. We use uppercase variables for lists, and $L_j$ denotes the $j$th item of
list $L$. Given a list of byte strings $L$ we denote its in-order concatenation by $\Vert L$. We use
$[]$ to denote the empty list. $\tf{Rand}(n)$ denotes a function that produces a uniformly random
string of $n$ bytes, and $\tf{Append}(L,x)$ returns a new list constructed by appending the element
$x$ to a list $L$.  
Given a byte string $x$, we use $x[i..j]$ to denote the substring from byte $i$ to byte $j$,
inclusive, with indexing beginning at $1$. We use the binary $\%$ operator to mean ``without the
prefix''; for example, $\text{abcd}\%\text{ab}=\text{cd}$. $\llbracket a, b \rrbracket$ denotes the
longest common prefix of byte strings $a$ and $b$. $\preceq (\prec)$ indicates (strict) string
prefixes. $|x|$ denotes the length of $x$ in bytes. $x \overset{{\scriptscriptstyle
\operatorname{R}}}{\leftarrow} S$ denotes sampling an element $x$ uniformly at random from finite
set $S$. 

We use a generic AEAD encryption scheme for our channel constructions in Sections \ref{sec:stream-construction} and \ref{sec:datagram-construction}, as well as certain counterexamples. 
\sloppy The scheme consists of a triple of algorithms $\tf{Gen}(\lambda)$, $\tf{Enc}_k(\tok{nonce},m,\tok{ad})$, $\tf{Dec}_k(\tok{nonce},c,\tok{ad})$, as defined by Rogaway~\cite{rogaway2002authenticated}, with both the AUTH and PRIV properties for integrity and confidentiality, respectively. In the remainder, we drop the associated data argument since it is not necessary for our work. We denote decryption errors with the distinguished error symbol $\errDec$. 
We assume that the scheme satisfies IND\$-CPA~\cite{rogaway2011evaluation}, meaning that its ciphertext outputs are indistinguishable from random bytes. 
We further assume that the scheme has a fixed nonce size, $\nonceSize$, a fixed authentication tag size $\tagSize$ and assume that the scheme has a fixed overhead associated with encryption in the form of an additive constant, which we refer to as $\aeadOverhead$ meaning that $|\tf{Enc}_k(\tok{nonce},m)| = |m| + \aeadOverhead$ for any valid key $k$, nonce $\tok{nonce}$, and message $m$. We call schemes with this property \emph{length additive}. 
We note that this property implies \emph{length regularity}; that is, input plaintexts of the same length produce output ciphertexts of the same length. Standard AEAD schemes, such as AES-GCM, are believed to satisfy all of the properties we assume.

\emph{Datastream} channels are intended to model the interface provided by TCP, where correctness
requires, and only requires, that the sent \emph{bytes} are all received and in the same order that
they were sent. Datastream channels are connection-based and can be \emph{closed} explicitly (\eg,
with a TCP FIN packet), observably terminating the connection. Messages sent in the datastream model
may be fragmented, with internal buffering behavior, network conditions, and adversarial
manipulation all affecting the size of the messages that are sent and received, which can be
distinct from the sizes of the messages explicitly passed to the channel interface at the level of
an application. \emph{Datagram} channels are intended to model the interface provided by the UDP and IP protocols, where messages can arrive out of order, fail to arrive at all, or arrive multiple times, and all messages have an explicit length.

Let $\mathcal{M} = \{0,1\}^*$ be the plaintext message space for a channel. We use $\bot$ as
another distinguished error symbol (\ie, $\bot\notin\mathcal{M}$), which will indicate an error in
the channel operation.

\emph{Traffic analysis}~\cite{raymond2001traffic} is used by a network adversary to infer sensitive
information about the metadata (\eg, the sender or receiver) or content of encrypted traffic. Modern
traffic analysis techniques rely on many features to make these inferences, including timing data,
plaintext message headers, protocol control messages, and message lengths. \emph{Traffic
shaping}~\cite{bocovich2016slitheen,wright2009traffic,barradas2017deltashaper,fifieldTordev}, used
in the fields of anonymous communication, network privacy, and censorship circumvention, frustrates
traffic analysis by changing the timing, number, and lengths of the messages.

\section{Data Stream Channels}
We present security notions for Fully Encrypted Protocols in the datastream setting using the model
of Fischlin et al.~\cite{DIAS}. We make two major additions to this model that are important in FEP
context. First, we allow the sender to indicate the desired length of the output ciphertext, which
enables \shaping{} for metadata protection. Second, we allow the receiver to close the channel
in response to an input ciphertext, which models this potential information leak.

\subsection{Channel Model}
A \emph{datastream channel} consists of three algorithms:
\begin{enumerate}
    \item $(\tok[S]{st}, \tok[R]{st}) \gets \tf{Init}(1^\lambda)$, which takes a security parameter $\lambda$ and returns an initial sender state $\tok[S]{st}$ and receiver state $\tok[R]{st}$.
    \item $(\tok[S]{st}', c) \gets \tf{Send}(\tok[S]{st},m,p,f)$ which takes a sender state
    $\tok[S]{st}$, a plaintext message $m$, an output length $p$, and a flush flag $f$, and returns an updated sender state $\tok[S]{st}'$ and a (possibly empty) ciphertext fragment $c$.
    \item $(\tok[R]{st}', m, \tok{cl}) \gets \tf{Recv}(\tok[R]{st},c)$ which takes a receiver state $\tok[R]{st}$ and a ciphertext fragment $c$, and produces an updated receiver state
    $\tok[R]{st}'$, a plaintext fragment $m$, and a flag $\tok{cl}$ indicating whether the channel is to be closed.
\end{enumerate}
These algorithms provide a unidirectional communication channel from the sender to the
receiver. To use a channel, $\tf{Init}$ is called to produce the initial states. States
$\tok[S]{st}$ and $\tok[R]{st}$ should be shared with the sender and receiver, respectively, using
an out-of-band process. Such states may, for example, include a shared symmetric key or
public/private keypairs for each party.

The sender calls $\tf{Send}$ to send a message $m\in\mathcal{M}$ to the receiver. Note that the
correctness requirement will not guarantee that the output ciphertext $c$ contains all of $m$ (\ie,
plaintext fragmentation is allowed), nor that it is even a full ciphertext (\ie, ciphertext
fragmentation is allowed). Thus the sender will be allowed to buffer inputs, such as for performance
reasons. Moreover, $c$ might in fact contain multiple ciphertexts, in the sense that $\tf{Recv}$ may
need to perform multiple decryption operations to obtain all the contained plaintext. The output length $p$ will be used to support traffic shaping. The flush flag
$f$, if set, forces the output of all messages provided as input to a $\tf{Send}$ call up to the
given call.

The receiver calls $\tf{Recv}$ on a received ciphertext to obtain the sent plaintext. Correctness
will require that $\tf{Recv}$ can take as its inputs a fragmentation or merging of
previous outputs of $\tf{Send}$. Thus the receiver will also be allowed to buffer inputs so that
it can produce the sent plaintext once sufficient ciphertext fragments are received.
The output message $m$ may contain errors (\ie, $m\in\{0,1,\bot\}^*$). If the close
flag $\tok{cl}$ is set, that indicates that the receiver actively closes the channel after the given
$\tf{Recv}$ call (\eg, by sending a TCP FIN).

\subsection{Correctness} \label{sec:stream-correctness}
A channel should be considered correct if it delivers the data from the sender to the receiver in
the absence of malicious interference. As given by Fischlin et al.~\cite{DIAS}, datastream
correctness requires that the plaintext message bytes produced by $\tf{Recv}$ should match the
message bytes given to $\tf{Send}$ if the ciphertext bytes are correctly delivered from $\tf{Send}$
to $\tf{Recv}$. This requirement tolerates arbitrary fragmentation of the plaintexts and
ciphertexts. However, we add to this correctness notion a requirement related to channel
closures. 

The inclusion of closures raises the possibility of a trivial correct channel that performs no
useful data transfer and instead just immediately closes the channel. However, we also want to allow
the possibility that the receiver does close the channel for a reason such as some plaintext command
being received or a limit on received data being reached. To rule out the former while allowing the
latter, we introduce the notion of a close function and use it to parameterize correctness. A close
function $\mathscr{C}$ is a randomized function that takes as input a sequence of channel operations
(\ie, function inputs and outputs starting with \tf{Init} and then including \tf{Send} and
\tf{Recv} calls) ending in a final \tf{Recv} call. It outputs a bit indicating if that last
call should set its close output. 
We use $\mathscr{C} \equiv 0$ as a default close behavior, \ie, the channel is never closed.

Our correctness requirement is as follows:
\begin{definition}
Let $S$ be a sequence of channel operations, starting with \tf{Init} and followed by (possibly
interleaved) calls to \tf{Send} and \tf{Recv}. Assume each \tf{Send} or \tf{Recv} call receives as input the state produced by its previous invocation (or the output of \tf{Init} on the first such call). Let $n_s$ be the number of \tf{Send} calls and $n_r$
be the number of \tf{Recv} calls. Let $M$, $P$, and $F$ be lists containing the inputs to the
\tf{Send} calls, $C$ be a list of the \tf{Send} outputs, $C'$ be a list of the \tf{Recv} inputs, and
$M'$ and $\tok{CL}$ be lists of the \tf{Recv} outputs. A channel satisfies \emph{correctness} with
respect to close function $\mathscr{C}$ if, for all such sequences $S$ where $\cat C' \preceq \cat
C$, the following properties hold:
\begin{enumerate}
\item \textbf{Stream Preservation}: $\cat M' \preceq \cat M$
\item \textbf{Flushing}: If $F_{n_s} = 1$ and $\cat C' = \cat C$, then $\cat M' = \cat M$. 
\item \textbf{Close Coherence}: For any $i \in [1..n_r]$, if $\tok[i]{CL} = 1$, then for any $j$ with $n_r \geq j > i$, $(M'_j, \tok[j]{CL}) = (\emptySt, 0)$. 
\item \textbf{Close Correctness}: If the last channel operation is \tf{Recv}, the
distribution of its close output, given its inputs, is identically distributed to that of
$\mathscr{C}$ on $S$.
\end{enumerate}
\end{definition}
The Stream Preservation property ensures that delivered output bytes are always some prefix of the
plaintext input bytes. The Flushing property enforces that if the flush flag is given, then all
the plaintext input bytes are output by \tf{Send}. Close Coherence ensures that a channel is closed only once and that \tf{Recv} outputs no messages afterwards.
Close Correctness enforces that the channel implements
the externally supplied close behavior $\mathscr{C}$. The first two requirements of our correctness definition coincide with the correctness definition of Fischlin \etal~\cite{DIAS}.

\subsection{Secure Close Functions} \label{sec:secure-closures}

While for correctness we make no assumptions about the desired channel-close behavior, for
security we will want to limit it in some ways. Arbitrary close behavior may leak secret
information. For example, a close function might close the channel if a certain plaintext byte
sequence is received, which would violate confidentiality.

We therefore define a \emph{secure close function}, which will prevent closures from leaking any
information beyond what is implied by the stream of sent bytes. We require that a secure close
function be able to be expressed as a function taking the following inputs: (1) $C$, a
concatenation of the list of the ciphertexts produced by \tf{Send} calls; (2) $C'$, the list of the
ciphertexts given as inputs to \tf{Recv} calls; (3) $\tok{CL}$, the list of the channel-close
outputs of \tf{Recv} calls; and (4) $c$, the final ciphertext input to \tf{Recv}. As with all close
functions, a secure close function returns a bit indicating if the channel should close. A secure
close function must also be a probabilistic polynomial-time function. With its limited inputs and
computational complexity, a secure close function not only prevents closures from leaking additional information beyond what is implied by the \tf{Send} outputs, it further restricts that leakage to that implied by their \emph{concatenation}. This choice prevents closures from revealing boundaries between \tf{Send} outputs, which can be hidden from real-world network adversaries due to fragmentation below the application layer. Note that,
although a passive adversary can observe the ordering of \tf{Send} and \tf{Recv} calls, we need not
give that information to a secure close function, as the receiver itself cannot observe that
ordering and thus could never realize any close function that depends on it. 

In addition to preventing the leakage of plaintext data, secure close functions also prevent
leaking certain metadata. For example, they preclude closing as soon as an error is detected in a
protocol with variable-length ciphertexts, which would reveal some of the internal structure of
the protocol messages, as two \tf{Send} outputs are indistinguishable to a secure close function
from a single \tf{Send} output of the same total length due to their concatenation in the input.
Such a behavior is frequently observed in real-world FEPs, such as obfs4 and one direction of
Shadowsocks (see \S~\ref{sec:analysis}).

Despite their restrictions, secure close functions can still express interesting and useful
behavior. For example, they include closures that occur after receiving some maximum amount of data.
They also allow for closures at some point after an error is introduced, for example, after a
modified byte is received and the total bytes received is a multiple of 1000 (\ie, the strategy
employed by the interMAC protocol~\cite{boldyreva2012security}).

Finally, secure close functions are
realistic. Both academic and real-world protocols have close behaviors which can, in
whole or in part, be realized by a secure close
function, such as the fixed-boundary closures of InterMAC~\cite{albrecht2019libintermac} and the
never-close behavior of Shadowsocks in one direction~\cite{shadowsocks}. We also observe a
real-world attempt, in the VMess protocol~\cite{vmess}, to obscure when ciphertext errors are
detected by randomizing the timing of a subsequent closure, which doesn't quite satisfy secure
closures and consequently leaks metadata (see \S~\ref{sec:analysis}). 

\subsection{Confidentiality and Integrity Definitions}
We adopt several datastream confidentiality and integrity definitions from Fischlin \etal~\cite{DIAS}. We use IND-CPFA/IND-CCFA for passive/active indistinguishability (\ie, against chosen plaintext/ciphertext fragment attacks). We use INT-PST/INT-CST for integrity against
plaintext/ciphertext stream manipulation. We use these notions without alteration beyond the new function signatures associated with our channel model (\ie, adding the \shaping parameter $p$ and close-flag $f$ to the inputs of \tf{Send} and outputs of \tf{Recv}, respectively).

We also introduce a new passive confidentiality notion that includes adversarial observation of
closures. The modified notion, IND-CPFA-CL, gives a receiving oracle to the adversary that only
returns channel closures. Moreover, as the adversary is passive, the definition requires that the
adversary correctly deliver to the receiving oracle the ciphertext bytestream output produced by the
sending oracle. IND-CPFA-CL is a confidentiality goal for FEPs, but we do use standard IND-CPFA as a
means to prove other properties. See
\iffull
    Appendix~\ref{app:dstream-sec}
\fi
for a precise definition of
IND-CPFA-CL.

\subsection{Fully Encrypted Datastream Protocols}
We introduce passive and active security notions for Fully Encrypted Protocols in the datastream
setting. The goal of these definitions is to ensure that protocols satisfying them do not reveal
protocol metadata through the observable bytes and channels closures. The passive definition is
\passiveObfs, or FEP security against a chosen plaintext fragment attack, and the active definition
is \activeObfs, or FEP security against a chosen ciphertext fragment attack. The active notion is
implicitly parameterized by a secure close function, and we use \activeObfs[$\mathscr{C}$] to explicitly indicate security with respect to the close function $\mathscr{C}$. The definition for both notions is as follows:
\begin{definition}\label{def:fep-ds-security}
    A channel satisfies \bothObfs{}, for $x\in \{\textrm{CPFA}, \textrm{CCFA}\}$ if, for any PPT adversary
    $\mathcal{A}$, $\left|P\left[\mathsf{Exp}^{\textsf{\bothObfs},b}_{\mathcal{A}}(1^{\lambda}) = 1 \big| b \overset{{\scriptscriptstyle \operatorname{R}}}{\leftarrow} \{0,1\}\right] - 1/2\right|$ is negligible in the security parameter $\lambda$.
\end{definition}

The security experiment used in this definition (Algorithm~\ref{alg:obfs-exp}) gives adaptive
access to a sending oracle (Algorithm~\ref{alg:activeObfs-send}). That oracle calls the \tf{Send}
function of the underlying channel and then returns to the adversary either the output or the same number of genuinely random bytes, depending on a secret random bit $b$. The adversary is challenged
to distinguish between observing the outputs of \tf{Send} and random byte strings of the same
lengths. This experiment captures a key goal of a passive FEP---that every byte sent should appear
random to the adversary.

In the active setting, the adversary is also given adaptive access to a receiving oracle
(Algorithm~\ref{alg:activeObfs-recv}). The behavior of the oracle depends on the secret bit $b$.
If $b=0$, then the oracle calls the channel \tf{Recv}. If the sending and receiving byte streams are
still in sync (\ie, no received byte differs from the byte in the same position in the sent byte stream), then only the close flag from \tf{Recv} is returned to the adversary. If those byte
streams are out of synchrony, then both the output message and the close flag are returned. If
the streams just become out of sync, then the in-sync and out-of-sync components are separated
before returning any message produced out-of-sync to the adversary (the logic largely follows the
similar oracle in the IND-CCFA definition of Fischlin et al.~\cite{DIAS}).
If $b=1$, the oracle simply returns the empty string
and the close flag prescribed by the close function $\mathscr{C}$.

For active security, the receiving oracle should not yield outputs distinguishing the $b=0$ and
$b=1$ cases. The $b=1$ case yields a simple behavior, which the $b=0$ case only differs from if a
non-empty message is ever output or the prescribed close behavior is not followed. The relevance
of non-empty outputs to the FEP goal is that they occur when modified messages yield valid outputs,
which may influence observable behavior of the receiver and thus reveal metadata about the integrity
features of the protocol. Because the prescribed close behavior is a secure close function,
conforming to it ensures that the closures leak no information beyond what is already revealed by
the ciphertext byte sequence. Thus, if passive security is already satisfied, active security
ensures that only the ciphertext lengths could potentially reveal protocol metadata.

\begin{algorithm}
    \caption{$\mathsf{Exp}^{\textsf{\bothObfs},b}_{\mathcal{A}}(1^{\lambda})$: FEP security experiment}    \label{alg:obfs-exp}

    \begin{algorithmic}[1]
        \State $(\tok[S]{st}, \tok[R]{st}) \gets \tf{Init}(1^\lambda)$
	\State $C_S, C_R, C_{\tok{CL}}, \tok{sync} \gets \emptyList, \emptyList, \emptyList, 1$
        \State $b' \gets \left\{
            \begin{array}{ll}
                \mathcal{A}^{\mathcal{O}^b_{\tf{Send}}()}(1^\lambda) & \textrm{if $x =$ CPFA}\\
                \mathcal{A}^{\mathcal{O}^b_{\tf{Send}}(), \mathcal{O}^b_{\tf{Recv}}()}(1^\lambda) & \textrm{if $x =$ CCFA}
            \end{array}\right.$

        \State \Return $b' = b$
    \end{algorithmic}
\end{algorithm}

\begin{algorithm}
    \caption{$\mathcal{O}^b_{\tf{Send}}(m,p,f)$: FEP sending oracle}
    	\label{alg:activeObfs-send}

    \begin{algorithmic}[1]
        \State $(\tok[S]{st}, c_0) \gets  \tf{Send}(\tok[S]{st},m,p,f)$
        \State $c_1 \gets \tf{Rand}(|c_0|)$
        \State $C_S \getsAppend{C_S}{c_b}$
        \State \Return $c_b$ to $\mathcal{A}$
    \end{algorithmic}
\end{algorithm}

\begin{algorithm}[h]
    \caption{$\mathcal{O}^b_{\tf{Recv}}(c)$: FEP receiving oracle}
        \label{alg:activeObfs-recv}

    \begin{algorithmic}[1]
        \If{$b=0$}
            \If{$\tok{sync} = 0$} \textrm{ // already out of sync with \textsc{Send}}
                \State $(\tok[R]{st}, m, \tok{cl}) \gets \tf{Recv}(\tok[R]{st},c)$
                \State \Return $(m, \tok{cl})$ to $\mathcal{A}$
            \ElsIf{$(\cat C_R) \Vert c \preceq (\cat C_S)$} \textrm{ // in sync with \textsc{Send}}
                \State $(\tok[R]{st}, m, \tok{cl}) \gets \tf{Recv}(\tok[R]{st},c)$
                \State $C_R \getsAppend{C_R}{c}$
                \State \Return $(\emptySt, \tok{cl})$ to $\mathcal{A}$
            \Else \textrm{ // either deviating or exceeding \textsc{Send} outputs}
                \If{$(\cat C_R) \prec \llbracket (\cat C_R) \Vert c, (\cat C_S) \rrbracket$} \textrm{ // partial sync}
                    \State $\widetilde{c} \gets \llbracket (\cat C_R) \Vert c, (\cat C_S) \rrbracket \textrm{ \% } (\cat C_R)$ \textrm{ // sync part}
                    \State $(\widetilde{\tok[R]{st}}, \widetilde{m}, \widetilde{\tok{cl}}) \gets \tf{Recv}(\tok[R]{st}, \widetilde{c})$
                    \State $(\tok[R]{st}, m, \tok{cl}) \gets \tf{Recv}(\tok[R]{st}, c)$
                    \State $m' \gets m \textrm{ \% } \llbracket m, \widetilde{m} \rrbracket$ \textrm{ // out-of-sync output}
                \Else \textrm{ // none of $c$ in sync with \textsc{Send} outputs}
                    \State $(\tok[R]{st}, m', \tok{cl}) \gets \tf{Recv}(\tok[R]{st},c)$
                \EndIf
                \If{$(\cat C_S) \npreceq (\cat C_R) \Vert c \textbf{ or } m'\neq\emptySt$}
                    \State $\tok{sync} \gets 0$
                \EndIf
                \State $C_R \getsAppend{C_R}{c}$
                \State \Return $(m', \tok{cl})$ to $\mathcal{A}$
            \EndIf
        \Else \textrm{ // $b=1$, return empty string and desired closure}
            \State $\tok{cl} \gets \mathscr{C}(\cat C_S, C_R, C_{\tok{CL}}, c)$ \textrm{ // produce close output}
            \State $C_R \getsAppend{C_R}{c}$
            \State $C_{\tok{CL}} \getsAppend{C_{\tok{CL}}}{\tok{cl}}$
            \State \Return $(\emptySt, \tok{cl})$ to $\mathcal{A}$
        \EndIf
    \end{algorithmic}
\end{algorithm}

\subsection{\shaping} \label{sec:obfs-defn-length-shaping}
The passive and active FEP security definitions essentially ensure that only the ciphertext lengths
can leak protocol metadata. Those lengths can also reveal message data, as the length of a message
may in some settings be related to its contents. To enable full metadata protection,
our channel definition allows that a requested length $p$ can be provided to each call to \tf{Send}.
We say that channels that provide the requested lengths satisfy \shaping{}:
\begin{definition} \label{def:length-shaping}
A channel satisfies \shaping{} if, for any state $\tok[S]{st}$, message $m$, and integer $p \geq
0$, with $(\tok[S]{st}, c)\gets \tf{Send}(\tok[S]{st}, m, p, f)$, if $f=0$ then $|c| = p$, otherwise $|c| \geq p$.
\end{definition}

Note that, in this definition, if the flush flag is set then the desired length may be exceeded,
which allows for the case that there are a large number of buffered message bytes that must be
flushed. The definition also only imposes a requirement for $p\ge 0$, which
allows channels freedom to implement alternate behaviors for negative $p$ values.

Traffic shaping enables protection of protocol metadata by setting the $p$ inputs to values that are
independent of the protocol being used. For example, the $p$ values could all be set to a constant,
which would hide the number and sizes of metadata fields as well as the length of the plaintext
messages. This is similar to the strategy used by the Tor protocol of putting all data into
fixed-size cells to prevent traffic analysis~\cite{tor-protocol-spec}.

\section{Relations Between Datastream Notions}
\label{sec:stream-relations}
In this section we present relations among our novel security notions and previous security notions for datastream channels.


While our novel FEP definitions are designed primarily to enforce that all output bytes appear indistinguishable from random, we observe in this section that these properties actually imply standard cryptographic datastream security properties. In particular, FEP-CCFA is a strong property that directly implies ciphertext stream integrity, and FEP-CPFA alongside channel length regularity (a minor property discussed below) imply passive confidentiality. Finally, we highlight as a primary conclusion from this section the observation from Figure~\ref{fig:stream-relations} that a channel satisfying FEP-CCFA for a secure close function and channel length regularity satisfies all other security properties we identify. 

\usetikzlibrary {positioning,shapes.geometric}
\begin{figure}
\resizebox{0.5\textwidth}{!}{
\begin{tikzpicture}[every edge/.style={draw, >={Stealth[length=4mm, width=2mm]}}]

\node(CH-REG) [ellipse,draw] at (5,0) {CH-REG};

\node(FEP-CCFA) [ellipse,draw] at (-2,1) {FEP-CCFA};
\node(FEP-CPFA) [ellipse,draw] at (3,1) {FEP-CPFA};
\node(IND-CPFA-CL) [ellipse,draw] at (8, 1) {IND-CPFA-CL};
\node(LRT) [blue,fill,circle,minimum size=.1cm,inner sep=0pt] at (6,1) {};

\node(INT-CST) [rectangle,draw] at (-2,-2) {INT-CST};
\node(ERR-PRED) [rectangle,draw] at (4,-2) {ERR-PRED};
\node(IND-CPFA) [rectangle,draw] at (8,-2) {IND-CPFA};

\node(IND-CCFA) [rectangle,draw] at (4,-3) {IND-CCFA};
\node(ERR-FREE) [ellipse,draw] at (4,-1) {ERR-FREE};
\node(INT-PST) [rectangle,draw] at (-2,-3) {INT-PST};
\node(GCT) [blue,fill,circle,draw,inner sep=0pt,minimum size=0.1cm] at (1,0) {};
\node(FPT) [fill,circle,minimum size=.1cm,inner sep=0pt] at (4,-2.35) {};

\path[thick]
	(INT-CST) edge[->] (INT-PST);

\path[thick]
	(INT-CST) edge (FPT);
\path[thick]
	(ERR-PRED) edge (FPT);
\path[thick]
	(IND-CPFA) edge (FPT);
\path[thick]
	(FPT) edge[->] (IND-CCFA);


\path[thick]
	(FEP-CCFA) edge[color=blue,dashed,->] node[left] {Thm \ref{prop:cst-ccfa}} (INT-CST);
\path[thick]
	(FEP-CCFA) edge[color=blue,dashed,->] node[above] {Trivial} (FEP-CPFA);

\path[thick]
    (IND-CPFA-CL) edge[color=blue,dashed,->,bend left] node[right] {Trivial} (IND-CPFA);
    
\path[thick]
    (CH-REG) edge[color=blue,dashed,->] (LRT);
\path[thick]
    (FEP-CPFA) edge[color=blue,dashed,->] (LRT);
\path[thick]
(LRT) edge[color=blue,dashed,->] node[above left] {Thm \ref{prop:lenreg-fep-ind}} (IND-CPFA-CL);
    
\path[thick,color=blue]
	(IND-CPFA) edge[dashed,bend left,->] node[pos=.5,left] {Thm \ref{prop:corr-cpfa-cl}} (IND-CPFA-CL);

\path[thick,color=blue]
	(FEP-CPFA) edge[dashed,bend left=10,->]  (GCT) 
	(INT-CST) edge[dashed,->] (GCT)
	(ERR-FREE) edge[dashed,->] (GCT)
	(GCT) edge[dashed,->,bend right=5] node[right,pos=.7] {Thm \ref{thm:new-general-comp}} (FEP-CCFA.south east);
	
\path[thick]
	(FEP-CCFA) edge[color=blue,dashed,->,bend right=10] node[below,pos=.75] {Trivial} (ERR-PRED);
\path[thick]
	(ERR-FREE) edge[color=blue,dashed,->] (ERR-PRED);

\end{tikzpicture}
}
\caption{Relations between notions for correct datastream channels that realize a secure close function $\mathcal{C}$. Rectangles and solid arrows are results from Fischlin~\etal~\cite{DIAS}; ellipses and dashed arrows are novel notions and relations} \label{fig:stream-relations}
\end{figure}

Figure \ref{fig:stream-relations} summarizes our results relating datastream security notions. First we include relevant notions and results from Fischlin~\etal~\cite{DIAS} with solid lines and rectangles, which establishes the result that the passive confidentiality notion (IND-CPFA) and the integrity of ciphertexts (INT-CST), together with a new notion ERR-PRED, which establishes the predictability of error symbols that will be produced by \tf{Recv}, imply their active confidentiality notion IND-CCFA. CH-REG refers to Datastream Channel Length Regularity (Definition \ref{def:length-regular-channel}). Since in our construction we do not use error symbols, we for convenience include the notion ERR-FREE, which refers to the property that \tf{Recv} never produces in-band error symbols and trivially implies ERR-PRED. 

Below we give theorem statements for each relation in Figure~\ref{fig:stream-relations}, with the proofs in
\iffull
	Appendix~\ref{appendix:stream-relations}.
\fi

While our notions do not imply confidentiality, this is only because of the issue of ciphertext lengths, and we formalize this intuition below. First, we define a length-regularity property for a channel, where we require that the length of the outputs of \textsc{Send} do not depend on the content of the messages. 

\begin{definition}\label{def:length-regular-channel} Let $M^0$ and $M^1$ be
$n$-length lists of messages such that, for all $i$, $|M^0_i| = |M^1_i|$. Let $P$ and $F$ be
an $n$-length integer sequence and an $n$-length boolean sequence. Let $(\tok[S]{st}^0, C^0_i) \gets
\tf{Send}(\tok[S]{st}^0, M^0_i, F_i, P_i)$ and $(\tok[S]{st}^1, C^1_i) \gets
\tf{Send}(\tok[S]{st}^1, M^1_i, F_i, P_i)$, where in both cases $\tf{Send}$ is initialized
with $\tf{Init}$ and is then called sequentially as $i=1..n$, updating its state with each call. A
datastream channel is \emph{length regular} if, for any such $M^0$, $M^1$, $F$ and $P$,
and for all $i$, $|C^0_i| = |C^1_i|$. 
\end{definition} 

Next, Theorem~\ref{prop:lenreg-fep-ind} shows that a length-regular channel
satisfying \passiveObfs provides confidentiality.

\begin{theorem} \label{prop:lenreg-fep-ind} Suppose that a channel satisfies \passiveObfs and further that the \textsc{Send} function is \textit{length regular} in the sense of Definition \ref{def:length-regular-channel}. Then that channel satisfies IND-CPFA. \end{theorem}



Theorem~\ref{prop:corr-cpfa-cl} shows that if a channel is correct and provides
IND-CPFA (the standard confidentiality notion), then it provides the similar
confidentiality notion with closures, IND-CPFA-CL.
\begin{theorem} \label{prop:corr-cpfa-cl} If a channel satisfies correctness for a given secure close function $\mathcal{C}$, and IND-CPFA, then it satisfies IND-CPFA-CL. 
\end{theorem}



We show in Theorem~\ref{prop:cst-ccfa} that \activeObfs by itself implies the 
strong integrity notion INT-CST.
\begin{theorem} \label{prop:cst-ccfa} If a channel satisfies \activeObfs, then it satisfies INT-CST. \end{theorem}


We give in Theorem~\ref{thm:new-general-comp} a general result showing
sufficient conditions for passive FEP security (\passiveObfs) to imply active
FEP security (\activeObfs{}). Note that the ERR-FREE condition means that
$\tf{recv}$ does not produce in-band errors.

\begin{theorem} \label{thm:new-general-comp} If a channel satisfies correctness for a given secure close function $\mathscr{C}$, \passiveObfs, ERR-FREE, INT-CST, then it satisfies FEP-CCFA-$\mathscr{C}$.  \end{theorem}

Finally, we give some negative results showing that, similar to IND\$-CPA in the atomic setting, our datastream FEP security notions do not imply and are not implied by confidentiality notions. 

We observe that \passiveObfs does not imply IND-CPFA.
Even if a channel satisfies \passiveObfs, it may still embed plaintext information in the length of ciphertext fragments it produces, violating confidentiality. We give an explicit counterexample in
\iffull
	Appendix \ref{appendix:noCPFA}.
\fi

We similarly observe that IND-CCFA and IND-CPFA do not imply \passiveObfs.
To show that neither of these datastream confidentiality notions implies \passiveObfs, we simply observe that the AEAD-based construction of Fischlin \etal~\cite{DIAS} satisfies both of them but includes unencrypted length fields to delimit ciphertext block boundaries, clearly failing to satisfy \passiveObfs.

\section{A Datastream Fully Encrypted Protocol}
In this section we give a construction for a Fully Encrypted Protocol in the datastream setting, and we prove that it satisfies all of the desired security properties. The key challenge in designing the construction is to simultaneously achieve \shaping, correctness, and security.

\subsection{The Construction}
\label{sec:stream-construction}

\begin{figure*}[!ht]

\caption{\label{fig:stream-construction}A Datastream Fully Encrypted Protocol}
\makebox[\textwidth][c]{

\fbox{
    \begin{minipage}[t]{.5\textwidth}
		$\tf{Init}(1^{\lambda})$:
        \begin{algorithmic}[1]
				\State $k \rgets \mathcal{K}$
                \State $\tok[S]{st} = (k, 0, \emptySt, \emptySt)$
                \State $\tok[R]{st} = (k, 0, \emptySt, 0)$
                \State \textbf{return} $(\tok[S]{st}, \tok[R]{st})$
        \end{algorithmic}
        \hrule
        \vspace{1mm}
        $\tf{Send}(\tok[S]{st}, m, p, f)$:
		\begin{algorithmic}[1]
            \State $(k, \tok{seqno}, \tok{buf}, \tok{obuf}) \gets \tok[S]{st}$
            \State \label{line:tcp-buf} $\tok{buf} \gets \tok{buf} \Vert m$

            \If{$(|\tok{obuf}| \geq p) \land ((f=0) \lor (\tok{buf} = \emptySt))$}
                \State $c \gets \tok{obuf}[1 .. \tf{Max}(p,f*|\tok{obuf}|)]$
                \State $\tok{obuf} \gets \tok{obuf} \% c$
                \State \textbf{return} $(\tok[S]{st}, c)$
            \EndIf
            \State $o \gets \tf{Min}(|\tok{buf}|, \tok{il})$ \Comment Plaintext Length
	    \State $\ell_p \gets 0$ \Comment L8-12 calculates $\ell_p$, $\ell_c$
            \State $\ell_c \gets |\tf{Enc}_k(\tok{seqno},0^{(2 + o)})|$
            \While{$(|\tok{obuf}| + \ell_c + \lenlen < p) \land (\ell_c \neq \tok{ol})$}
                \State $\ell_p \gets \ell_p + 1$
                \State $\ell_c \gets |\tf{Enc}_k(\tok{seqno},0^{(2 + o + \ell_p)})|$
            \EndWhile
            \State $c \gets \tf{Enc}_k(\tok{seqno}, \ell_c)$ \Comment Length Block
            \State $\tok{seqno} \gets \tok{seqno} + 1$
            \State $m' \gets \ell_p \Vert (0^{\ell_p}) \Vert \tok{buf}[1..o]$
                \State $\tok{buf} \gets \tok{buf} \% \tok{buf}[1..o]$
                \State $c \gets c \Vert \tf{Enc}_k(\tok{seqno}, m')$ \Comment Payload Block
                \State $\tok{seqno} \gets \tok{seqno} + 1$
            \State $\tok{obuf} \gets \tok{obuf} \Vert c$
            \State \textbf{return} $\tf{Send}(\tok[S]{st},\emptySt, p, f)$

		\end{algorithmic}
	\end{minipage}
    \begin{minipage}[t]{.5\textwidth}
        $\tf{Recv}(\tok[R]{st},c)$:
        \begin{algorithmic}[1]
        \State $(k, \tok{seqno}, \tok{buf}, \tok{fail}) \gets \tok[R]{st}$
        \If{$\tok{fail} = 1$}
            \State \textbf{return} $(\tok[R]{st}, \emptySt, 0)$
        \EndIf
        \State $\tok{buf} \gets \tok{buf} \Vert c$
        \State $m \gets \emptySt$
        \While{$|\tok{buf}| \geq \lenlen$}
        \State $\ell_c \gets \tf{Dec}_k(\tok{seqno},\tok{buf}[1  .. \lenlen])$ \Comment Decrypt Length Block
        \If{$\ell_c  = \bot_\textsc{Dec}$}
                \State $\tok{fail} \gets 1$
                \State \textbf{return} $(\tok[R]{st}, \emptySt, 0)$
        \ElsIf{$|\tok{buf}| \geq  \lenlen + \ell_c$}
        \State $c' \gets \tok{buf}[\lenlen+1 .. \lenlen+\ell_c]$
        \State $\tok{buf} \gets \tok{buf} \% \tok{buf}[1 .. 1+\lenlen + \ell_c]$
                \State $m' \gets \tf{Dec}_k(\tok{seqno}+1,c')$ \Comment Decrypt Payload Block
                \State $\tok{seqno} \gets \tok{seqno} + 2$

                    \If{$m' = \bot_\textsc{Dec}$}
                    \State $\tok{fail} \gets 1$
                    \State \textbf{return} $(\tok[R]{st}, \emptySt, 0)$
                    \EndIf
                    \State $\ell_p \gets \tf{min}(m'[1,2],|m'|-2)$ \Comment Calculate Padding Length
                    \State $m \gets m \Vert m'[3 + \ell_p..]$
         \Else
         \State $\textbf{break}$
            \EndIf
        \EndWhile
        \State \textbf{return} $(\tok[R]{st}, m,0)$

	\end{algorithmic}
	\end{minipage}
}
}
\end{figure*}
We give a construction that satisfies all our security notions for the secure close function
$\mathcal{C} \equiv 0$, presented in Figure \ref{fig:stream-construction}. Our approach is inspired by Shadowsocks, and is designed around producing pairs of ciphertexts, the first of which is a ``length block'' which has a fixed length, denoted $\lenlen = 2 + \tagSize$, and contains the length of the subsequent ``payload block'', which is limited to a length that can be represented within two bytes. The payload block also contains a two byte padding length field to denote internal (plaintext) padding. We denote the largest payload block length as $\tok{ol} = 2^{16}-1$ and set $\tok{il}$ as the largest plaintext length that can fit within a payload block ($2^{16} - 3 - \tagSize$). 

The \tf{Send} function contains an input buffer $\tok{buf}$ for plaintext and produces an output buffer $\tok{obuf}$ of ciphertext block pairs, outputting fragments from this buffer as requested by the caller. When \tf{Send} is called, it first checks to determine whether it can return the required number of bytes (determined by the conditional in Line 3 of \tf{Send}, which checks if $p$ output bytes are available and $f=0$, or if $f=1$ and the plaintext buffer is empty). If possible, \tf{Send} outputs the appropriate number of bytes from the output buffer ($p$ or $|\tok{obuf}|$, depending on $f$). Otherwise, \tf{Send} constructs a pair of ciphertext blocks according to our scheme in Lines 7--18 and calls itself which will result in termination or a further extension of $\tok{obuf}$ and another recursive call. The ciphertexts are constructed by encrypting as much of $\tok{buf}$ as possible, represented by the variable $o$, and then filling in padding bytes if needed so that the block pair has size $p$ or is the maximum size of $\lenlen + \tok{ol}$. 

The \tf{Recv} function keeps a buffer of received ciphertext bytes and waits to receive $\lenlen$ bytes, then decrypts and parses the length block to recover the length of the subsequent payload block. It then waits until the full payload block has been received before decrypting and returning the plaintext, stripping off the padding length and any padding bytes. \tf{Recv} checks for decryption errors and enters a fail state if either block type fails to decrypt. Our construction does not produce channel closures or error symbols.

We note that our construction generalizes in some ways.
Our buffering and padding approach can be easily adapted for other FEP designs, such as obfs4~\cite{obfs4} and InterMAC~\cite{boldyreva2012security}, so that they satisfy \shaping. Further, although
timing is outside of our model, an implementation could close the connection based on a timeout,
where the timeout must depend only on how much time has elapsed since the last ciphertext fragment was
received.

\subsection{Properties}

We show that our construction satisfies all of the desired properties in Theorem \ref{thm:datastream-construction-all}, with the proofs contained in
\iffull
    Appendix \ref{appendix:datastream-construction-proofs}.
\fi
Our general approach is to first show the channel satisfies the three non-cryptographic properties: correctness, \shaping, and channel length regularity. We then establish individually each property required to apply Theorem \ref{thm:new-general-comp}: FEP-CPFA, INT-CST, and ERR-FREE (which is trivial). Theorem \ref{thm:new-general-comp} then implies the channel satisfies FEP-CCFA, which (along with channel length regularity) implies IND-CCFA, so we arrive at all of the desired datastream channel properties established in the previous sections. 

\begin{theorem}
	\label{thm:datastream-construction-all} If the AEAD scheme (\tf{Gen}, \tf{Enc}, \tf{Dec}) in the construction in Figure \ref{fig:stream-construction} satisfies IND\$-CPA, INT-CTXT, and is length additive (see Sec \ref{sec:aead-scheme}) then the channel construction satisfies correctness for the secure close function $\mathscr{C} \equiv 0$, \shaping, FEP-CCFA,  INT-CST, and IND-CCFA. 
\end{theorem}



\section{Fully Encrypted Datagram Transport Protocols}
\label{sec:datagram-channel-model}
In the datagram setting, messages are transmitted atomically (i.e. without fragmentation). Implied
in this property is that the length of each message can be determined by the receiver, unlike in the
datastream model where messages can be arbitrarily fragmented or merged. Even in benign
circumstances, messages in the datagram model may be delivered out of order or dropped. This setting
models the UDP transport protocol.

Let $\inLength$ be the maximum number of bytes in an input message. This number may vary with
the channel (i.e. with the protocol), but it is assumed that $\inLength > 0$. Let
$\messageSpace = \cup_{0\le i\le \inLength} \{0,1\}^{8i}$ be the space of
valid input messages. Thus, any message $m$ with $|m| \le \inLength$ is a valid input while
sending. Let $\emptySt$ be the ``empty''
message, that is, the message of length zero. Note that
$\emptySt\in\messageSpace$. We assume that there is a maximum number of bytes
$\outLength$ in deliverable outputs, which is typically larger than $\inLength$.
$\outLength$ does not vary with the channel and models the limitations of the underlying
delivery mechanism (\eg, the maximum size of a UDP datagram). Let $\bot$ be a distinguished error
symbol (i.e., an item distinct from other outputs), which can be produced when sending or receiving
a message, such as when the message to send is too long or the received message is malformed.

We also introduce the distinguished symbol $\nullsym$ to indicate a ``null'' message. This symbol is
different from the empty message $\emptySt$ in that it indicates no message rather than an empty
message. We introduce this symbol to distinguish ``chaff'' messages, which are those sent purely to
fool a network observer, from messages that are intended to be delivered to the receiver. A special
case of this is when the sender requests an output length that is too small to contain an
authentication tag, which we want to support to provide \shaping. $\nullsym$ can be used as both an
input while sending and an output while receiving, and so to avoid ambiguity it is not in
$\messageSpace$. $\nullsym$ can
be used as an input to produce a chaff message. If it is produced as an output while receiving,
the receiver should react as if no message was received.
This design choice is driven by our desire to ensure that if the channel operates normally, the receiver does not produce an error symbol. This combined with the requirement that messages can be produced of arbitrary size by an honest sender (because of \shaping), and the requirement that messages may be dropped or re-arranged without error means that in certain cases, malicious messages or errors introduced by an active adversary will be indistinguishable from valid chaff messages, which required changes to the $\tf{Recv}$ oracle in our security definitions, presented in Algorithm \ref{alg:fepdgram-recv}. 	

The datagram model offers different challenges to defining and achieving Fully Encrypted Protocols.
Atomic messaging avoids the necessity of communicating message length. However, the lack of delivery
guarantees means that each message must be individually decryptable, which, for example, rules out
using certain block-cipher modes across messages. At the same time, however, the \shaping
requirement will require sending messages so small they cannot even contain authentication tags, as
mentioned.

\subsection{Channel Model}
The datagram channel model consists of the following algorithms, which provide a unidirectional channel:
\begin{enumerate}
    \item $(\tok[R]{st}, \tok[S]{st}) \gets \tf{Init}(1^\lambda)$, which takes a security parameter
    $\lambda$ and generates the sender and receiver state. 
    \item $(\tok[S]{st}', c) \gets \tf{Send}(\tok[S]{st}, m, p)$, which takes a sender state and a plaintext message $m$, and a desired output length $p$. Its output $\tok[S]{st}'$ is an updated sender state, and its output $c$ is either a ciphertext or an error.
    \item $(\tok[R]{st}', m) \gets \tf{Recv}(\tok[R]{st}, c)$, which takes a receiver state and a ciphertext $c$. Its output $\tok[R]{st}'$ is an updated receiver state, and its output $m$ is either a plaintext message or an error.
\end{enumerate}

\subsection{Correctness}
The correctness requirement differs significantly from the datastream setting due to the lack of
reliable in-order delivery. Instead, the requirement should simply enforce that a message $m$ sent
through $\tf{Send}$ and then $\tf{Recv}$ produces $m$ again. This basic idea is complicated by the
possibility that $\tf{Send}$ may yield an error. We allow for such errors in the following
cases: (1) $m$ is too large (i.e., $|m| > \inLength$), (2) $p$ is too small for
$m$, and (3) $p$ is too large  (i.e., $p > \outLength)$.
Definition~\ref{def:datagram-correctness} gives the correctness requirement for a datagram channel.

\begin{definition} \label{def:datagram-correctness}
Fix any sequence of datagram-channel operations where the first call is to \tf{Init}, the remaining are (possibly interleaved) calls to $\tf{Send}$ and $\tf{Recv}$, and the sender (receiver) state outputs are correctly provided as inputs to the subsequent $\tf{Send}$ ($\tf{Recv}$) call.

The datagram channel satisfies \emph{correctness} if, for all such sequences, the following properties hold:
\begin{enumerate}
\item \textbf{Message Acceptance}: For every $m \in \messageSpace \cup \{\nullsym\}$, there exists some $p_m \leq \outLength$ such that with $(\tok[S]{st}, c) \gets \tf{Send}(\tok[S]{st}, m,p)$, if $p \geq p_m$ or $p < 0$, $c \neq \bot$. 
\item \textbf{Message Length}: For every $m \in \messageSpace \cup \{\nullsym\}$, with $(\tok[S]{st}, c) \gets \tf{Send}(\tok[S]{st},m,p)$, if $c \neq \bot$, $|c| \leq \outLength$.
\item \textbf{Message Delivery}: For any $\tf{Send}$ call with input $m\in\messageSpace\cup\{\nullsym\}$ and output $(\tok[S]{st}, c)$, where $c \neq\bot$, any subsequent $\tf{Recv}$ call with input $c$ must have output either $m$ or $\bot$, and the first such call must have output $m$.

\end{enumerate}
\end{definition}

\subsection{\shaping}
As with datastream channels, we would like our datagram-based protocol to support \shaping{} by
changing the length of its output messages. The desired output length is specified by the
$\tf{Send}$ parameter $p$. We require \tf{Send} to output $p$ bytes as long as $p$ is large enough
to accommodate the desired message and does not exceed the maximum datagram size $\outLength$.
In particular, for the null message $\nullsym$, values $0\le p\le \outLength$ should yield an 
output ciphertext of length $p$ ($c$ might be meaningless randomness if $p$ is too small to
accommodate ciphertext metadata such as a tag or nonce).
Definition~\ref{def:length-shaping-datagram} gives our precise \emph{\shaping} notion in the
datagram context.

\begin{definition} \label{def:length-shaping-datagram}
A datagram channel satisfies \emph{\shaping} if, for any state $\tok[S]{st}$ produced by
$\tf{Init}$ or a subsequent $\tf{Send}$ call, any message $m \in \messageSpace \cup\nullsym$, and any
integer $p
\geq 0$, the following hold for $(\tok[S]{st}', c) \gets \tf{Send}(\tok[S]{st}, m, p)$:
\begin{enumerate}
    \item If $c\neq \bot$ and $p \geq 0$, then $|c| = p$, and
    \item If $m=\nullsym$ and $p \leq \outLength$, then $c\neq \bot$.
\end{enumerate}
\end{definition}
Note that this definition places no constraints on the channel output when $p$ is negative.
Therefore, a client may turn off \shaping by using $p<0$.

\subsection{Security Definitions}
Providing security definitions for fully encrypted protocols is simpler for datagram protocols than
for datastream protocols, largely because the correctness requirement does not require tolerating
plaintext and ciphertext fragmentation. As with datastream protocols, our new
definitions for FEPs are distinct from existing security notions for confidentiality and integrity, and we therefore present and discuss those existing notions as well. All our security notions are
with respect to a probabilistic polynomial-time (PPT) adversary.

\subsubsection{Confidentiality and Integrity definitions}
We adopt confidentiality and integrity definitions from Bellare and
Namprempre~\cite{bellare2008authenticated}. We require only
slight modifications to adapt them to our stateful communication channel. Details appear in
\iffull
    Appendix~\ref{app:dgram-sec}.
\fi

For confidentiality, we use the \indcpa{} and \indcca{} definitions, which provide security with
respect to a passive and active adversary, respectively.
In both definitions, the adversary
is given adaptive access to a sending oracle, which makes use of the $\tf{Send}$ function of the underlying
channel. In \indcca{} security, the adversary also has
adaptive access to a receiving oracle, which uses the channel's $\tf{Recv}$ function.
The sending oracle requires the adversary to submit pairs of equal-length messages, and it outputs the encryption of one of them, depending on a secret random bit $b$. The receiving oracle only returns a decryption of the given ciphertext if it was not an
output of the sending oracle and if the decryption is neither $\bot$ nor $\nullsym$. The channel
provides confidentiality if the adversary cannot guess the secret bit $b$ with probability
non-negligibly different than random guessing.

For integrity, we use \intctxt{}, which challenges the adversary to produce an unseen ciphertext
that successfully decrypts. The adversary is given adaptive access to a sending oracle and a
receiving oracle, which make use of the channel's $\tf{Send}$ and $\tf{Recv}$ functions,
respectively. The channel provides integrity if the adversary cannot send the receiving oracle a
ciphertext that was not produced by the sending oracle but does successfully decrypt to a non-null output, except with negligible probability.

\subsubsection{Fully Encrypted Datagram Protocols}

Our novel security definitions in the datagram setting are \fepcpa{} and \fepcca{}, which define
security for Fully Encrypted Protocols against a chosen plaintext attack (i.e. a passive adversary)
and against a chosen ciphertext attack (i.e. an active adversary), respectively:
\begin{definition}\label{def:fep-dg-security}
    A channel satisfies \bothfepdgram{}, for $x\in \{\textrm{CPA}, \textrm{CCA}\}$ if, for a security parameter $\lambda$ and PPT adversary
    $\mathcal{A}$, $\left|P\left[\mathsf{Exp}^{\textsf{\bothfepdgram},b}_{\mathcal{A}}(1^{\lambda}) = 1 \big| b \overset{{\scriptscriptstyle \operatorname{R}}}{\leftarrow} \{0,1\}\right] - 1/2\right|$ is negligible in $\lambda$.
\end{definition}

In the related security experiment (Algorithm~\ref{alg:fep-dgram}), the adversary is given adaptive
access to a sending oracle (Algorithm~\ref{alg:fepdgram-send}). That oracle either faithfully
returns the output of the channel $\tf{Send}$ operation or replaces that output with the same number
of uniformly random bytes, depending on a randomly selected bit $b$. The adversary is thus
challenged to distinguish between the outputs of $\tf{Send}$ and random messages of the same
length.

In the active security definition, \fepcca{}, adaptive access to a receiving oracle is added
(Algorithm~\ref{alg:fepdgram-recv}). That oracle also depends on the secret bit $b$. If $b=0$, it
returns the output of $\tf{Recv}$ called on the given ciphertext, except if the
ciphertext was an output of $\tf{Send}$ or if it fails to decrypt to a non-null plaintext. If $b=1$,
the oracle always returns $\bot$. The active adversary can thus also attempt to get $\tf{Recv}$ to
produce anything but $\bot$, which can only happen if $b=0$ and if the adversary submits a novel,
valid, and non-null ciphertext. This definition models potential information leaks through
observable behavior of the recipient that would receive those outputs. Note that null outputs are
excepted because, in general, the output of $\tf{Send}$ may be required to be of a size too small to
contain an authentication tag, making small outputs forgeable.

\begin{algorithm}
    \caption{$\mathsf{Exp}^{\textsf{\bothfepdgram},b}_{\mathcal{A}}(1^{\lambda})$} \label{alg:fep-dgram}
    \begin{algorithmic}[1]
        \State $(\tok[S]{st}, \tok[R]{st}) \gets \tf{Init}(1^\lambda)$
        \State $C \gets \emptyset$
        \State $b' \gets \left\{
            \begin{array}{ll}
                \mathcal{A}^{\mathcal{O}^b_{\tf{\fepdgramsend}}()}(1^\lambda) & \textrm{if $x = $ CPA}\\
                \mathcal{A}^{\mathcal{O}^b_{\tf{\fepdgramsend}}(), \mathcal{O}^b_{\tf{\fepdgramrecv}}()}(1^\lambda) & \textrm{if $x =$ CCA}
            \end{array}\right.$

        \State \Return $b' = b$
    \end{algorithmic}
\end{algorithm}

\begin{algorithm}[H]
    \caption{$\mathcal{O}^b_{\tf{\fepdgramsend}}(m,p)$} \label{alg:fepdgram-send}
    \begin{algorithmic}[1]
        \State $(\tok[S]{st}, c_0) \gets  \tf{Send}(\tok[S]{st},m,p)$
        \If{$c_0 = \bot$}
        \State \Return $\bot$
        \EndIf
        \State $c_1 \gets \tf{Rand}(|c_0|)$

        \State $C \gets C \cup \{c_b\}$
        \State \Return $c_b$
    \end{algorithmic}
\end{algorithm}

\begin{algorithm}[H]
    \caption{$\mathcal{O}^b_{\tf{\fepdgramrecv}}(c)$} \label{alg:fepdgram-recv}
    \begin{algorithmic}[1]
        \If{$b = 0$}
            \State $(\tok[R]{st}, m) \gets \tf{Recv}(\tok[R]{st}, c)$
            \If{$c \notin C \land m \neq \bot \land m \neq \nullsym$}
                \State \textbf{return} $m$
            \Else
                \State \textbf{return} $\bot$
            \EndIf

        \EndIf
        \State \textbf{return} $\bot$
        
    \end{algorithmic}
\end{algorithm}

\section{Relations Between Datagram Notions}
\label{sec:datagram-relations}
In this section we present relations between our security notions for datagram channels in Theorems \ref{prop:datagram-cca-integrity}, \ref{prop:datagram-composition}, and \ref{prop:datagram-fepcca-indcca}. Proofs of these relations appear in
\iffull
  Appendix~\ref{appendix:datagram-relations-proofs}.
\fi
As with the datastream notions, our datagram  security definitions do not imply and are not implied by confidentiality properties because plaintext values may be leaked via the ciphertext lengths.

\begin{theorem} \label{prop:datagram-cca-integrity} If a Datagram channel satisfies FEP-CCA, it satisfies INT-CTXT. \end{theorem}

\begin{theorem} \label{prop:datagram-composition} If a Datagram channel satisfies FEP-CPA and INT-CTXT, it satisfies FEP-CCA. 
\end{theorem}

To obtain confidentiality from our datagram FEP notions, we require a property that ensures message lengths do not leak information about the plaintext content. We define Datagram Channel Length Regularity in Definition \ref{def:length-regularity-datagram}. 

\begin{definition}\label{def:length-regularity-datagram} Let $M^0$ and $M^1$ be
$n$-length sequences of elements in $\messageSpace \cup \{\nullsym\}$ such that, for all $i$,
either $|M^0_i| = |M^1_i|$ or $M^0_i = M^1_i = \nullsym$. Let $P$ be
an $n$-length integer sequence. Let  $(\tok[S]{st}^0, C^0_i) \gets
\tf{Send}(\tok[S]{st}^0, P_i, M^0_i)$ and $(\tok[S]{st}^1, c^1_i) \gets
\tf{Send}(\tok[S]{st}^1, P_i, M^1_i)$, where in both cases $\tf{Send}$ is initialized
with $\tf{Init}$ and is then called sequentially as $i=1..n$, updating its state with each call. A
datagram channel is \emph{length regular} if, for any such $M^0$, $M^1$, and $P$,
and for all $i$, either $|C^0_i| = |C^1_i|$ or $C^0_i = C^1_i = \bot$.
\end{definition} 

\begin{theorem} \label{prop:datagram-fepcca-indcca} If a Datagram channel satisfies FEP-CCA and Datagram Channel Length Regularity, it satisfies IND-CCA. \end{theorem}

\section{A Fully Encrypted Datagram Construction}
In the following section we introduce our datagram construction of a Fully Encrypted Protocol, and we prove it satisfies satisfies correctness, the desired security properties, and \shaping.

\subsection{The Construction}
\label{sec:datagram-construction}

We give a simple construction of a protocol that satisfies our two obfuscation security definitions and \shaping for datagram channels. To satisfy most of our security and obfuscation properties, we can simply use atomic encryption and decryption since datagram messages are atomic by construction. However, the \shaping property presents a challenge: not all lengths are possible outputs of an AEAD scheme, and since the messages are not ordered, they must each include a nonce as well. 

In our construction, we assume the same AEAD scheme from Section \ref{sec:aead-scheme}. For simplicity, we further assume that \tf{Enc} includes in its outputs each nonce as a prefix of the ciphertext when it is called, and that \tf{Dec} extracts the first $\nonceSize$ bytes from its input as the nonce to use for decryption. Let $\aeadOverhead = \nonceSize + \tagSize$ represent the encryption overhead since our construction uses fresh nonces for each ciphertext. We give an explicit encoding scheme for our message space into bytes as follows: we encode the null message $\nullsym$ as a single zero byte, and prefix all other messages (including the empty string $\emptySt$) with a leading $1$ byte. We set $\outLength = 65507$ (the maximum size UDP payload length), and we include in our message space all byte strings of length up to $\inLength = \outLength - \aeadOverhead - 3$, reserving two bytes for the plaintext length and one for the message type. We denote the length of a ciphertext containing the null message as $\nullLength = 1 + \aeadOverhead$. We abuse notation and write $|m|$ to mean both the size of $m$ and the two byte unsigned integer that represents that size, depending on the context.

Our protocol is stateless, so each function simply returns the symmetric key $k$ as the new state after each call. The \tf{Send} function generates a fresh nonce, and if \shaping is off, directly applies \tf{Enc} to produce the datagram. If the desired output length $p$ is too small for an authenticated message, we produce random bytes (Lines 6--7). We produce arbitrary length chaff messages (Line 9) or padded plaintext messages (Line 13), and return with an error if $p$ is too large to fit in a datagram or $m$ is too large to fit within $p\geq 0$ (Line 10). \tf{Recv} interprets all small messages as chaff, and then decrypts larger messages, returning the plaintext, chaff symbol, or error where appropriate. 

\begin{figure*}[ht!]

\makebox[\textwidth][c]{

\fbox{
    \begin{minipage}[t]{.55\textwidth}
        $\tf{Send}(k, m, p)$:
		\begin{algorithmic}[1]
		    \State $\tok{nonce} = \tf{Rand}(\nonceSize)$

            	\If{$p < 0$ and $m = \nullsym$}
            			\State \textbf{return} $(k, \tf{Enc}_k(\tok{nonce},0))$
            	\EndIf
            		\If{$p < 0$ and $|m| \leq \inLength$}
            			\State \textbf{return} $(k, \tf{Enc}_k(\tok{nonce},1\Vert |m| \Vert m))$
            		\EndIf

        		\If{$m = \nullsym$ and $p < \nullLength$}
            \State \textbf{return} $(k, \tf{Rand}(p))$
			\EndIf
            
            \If{$m = \nullsym$ and  $\nullLength \leq p \leq \outLength$}
            \State \textbf{return} $(k,\tf{Enc}_k(\tok{nonce},0\Vert 0^{p - \nullLength}))$
            \EndIf
            \If{$p > \ell_{\tok{out}}$ or $|m| >  p - \aeadOverhead - 3$}
            \State \textbf{return} $(k,\bot)$
            \EndIf

			\State $\tok{pt} = 1\Vert |m|\Vert 0^{(p - |m| - \aeadOverhead - 3)}\Vert m$
            \State \textbf{return} $(k,\tf{Enc}_k(\tok{nonce},\tok{pt}))$ 

		\end{algorithmic}

	\end{minipage}
    \begin{minipage}[t]{.45\textwidth}
        $\tf{Recv}(k,c)$:
		\begin{algorithmic}[1]
		\If{$|c| < \nullLength$}
		\State \textbf{return} $\nullsym$
		\EndIf
		\State $\tok{buf} \gets \tf{Dec}_k(c)$
		\If{$\tok{buf} = \errDec$}
		\State \textbf{return} $\bot$
		\EndIf
		\If{$\tok{buf}[1] = 0$}
		\State \textbf{return} $\nullsym$
		\EndIf
		\State $i \gets |\tok{buf}| - \tok{buf}[2,3] + 1$ \Comment Two-byte unsigned int
		\State \textbf{return} $(k,\tok{buf}[i..])$
		\end{algorithmic}
		\hrule
		\vspace{.5mm}
		$\tf{Init}(1^{\lambda})$:
        \begin{algorithmic}[1]
		\State $k \gets \tf{Gen}(1^{\lambda})$
        \State \textbf{return} $(k, k)$
        \end{algorithmic}

	\end{minipage}
}
}

\caption{\label{fig:construction-datagram}Our datagram channel construction}
\end{figure*}

\subsection{Properties}

We prove in Theorem \ref{thm:datagram-construction-all} that our datagram construction satisfies all of the desired properties presented in Section \ref{sec:datagram-channel-model}. The proofs are in
\iffull
	Appendix \ref{appendix:datagram-construction-proofs}.
\fi
Our approach is to establish correctness, \shaping, datagram channel length regularity, FEP-CPA, and INT-CTXT, and then apply the theorems from Section \ref{sec:datagram-relations} to establish the remaining properties. 

\begin{theorem}
	\label{thm:datagram-construction-all} If the AEAD scheme (\tf{Gen}, \tf{Enc}, \tf{Dec}) in the channel construction in Figure \ref{fig:construction-datagram} satisfies IND\$-CPA, INT-CTXT, and length additivity (see Sec \ref{sec:aead-scheme}) then the channel construction satisfies correctness for Datagram channels, \shaping, FEP-CCA, INT-CTXT, and IND-CCA. 
\end{theorem}

\section{Existing Fully Encrypted Protocols}
\label{sec:analysis}
\newcommand{\greencheck}{{\color{green} \checkmark}}
\newcommand{\redx}{{\color{red} X}}
\newcommand{\yellowcirc}{{\color{orange} $\sim$}}
\newcommand{\keepalive}{\dagger}
\newcommand{\serverspecific}{\textsc{s}}
\begin{table*}
	\resizebox{\textwidth}{!}{%
\begin{tabular}{|l|c|c|c|c|l|}\hline
\textbf{Datastream Protocol} & \textbf{Close Behavior} & \textbf{FEP-CPFA} & \textbf{FEP-CCFA} & \textbf{Length Obfuscation} & \textbf{Min Size}\\\hline\hline
Shadowsocks-libev (request) & Never &\greencheck & $\greencheck^*$ & None & $35$\\\hline
Shadowsocks-libev (response) &  Auth Fail &\greencheck & \redx & None & $35$\\\hline
V2Ray-Shadowsocks (request) & $\text{Drain}^*$ &\greencheck & \redx & None & $35$\\\hline
V2Ray-Shadowsocks (response) & Auth Fail &\greencheck & \redx & None & $35$\\\hline
V2Ray-VMess & $\text{Drain}^*$ & \greencheck & \redx & Padding & $18^{\keepalive}$ \\\hline
Obfs4 & Auth Fail & \greencheck &\redx & Padding & $44^*$\\\hline
OpenVPN-XOR & Auth Fail & \redx &\redx &None & $42^{\keepalive}$\\\hline
Obfuscated OpenSSH & Auth Fail & \redx & \redx &None& 16\\\hline
Obfuscated OpenSSH-PSK & Auth Fail & \greencheck & \redx & None & 16\\\hline
kcptun & Never & \greencheck & \redx & None & $52^{\keepalive}$\\\hline
Construction (Sec \ref{sec:stream-construction}) & Never & \greencheck & \greencheck & \shaping & 1\\\hline\hline
\textbf{Datagram Protocol} &  & \textbf{FEP-CPA} & \textbf{FEP-CCA} & \textbf{Length Obfuscation} & \textbf{Min Size}\\\hline\hline
Shadowsocks-libev & &  \greencheck & \greencheck & None & $55^*$\\\hline
Wireguard-SWGP (paranoid) & & \greencheck & \greencheck & Padding & $75^{\keepalive *}$\\\hline
OpenVPN-XOR & & \redx & \redx & None & $40^{\keepalive}$\\\hline
Construction (Sec \ref{sec:datagram-construction}) & & \greencheck & \greencheck & \shaping & 0\\\hline
\end{tabular}%
}
\caption{Our definitions applied to FEPs. *See discussion for significant nuances. $\dagger$ The smallest message is a keepalive. \greencheck\xspace We have strong evidence the protocol satisfies the definition from documentation, code, and our experiments (properties are proven for our constructions). \redx\xspace We have demonstrated by experiment that the definition is not satisfied, or that fact is clear from the protocol design.}
\label{tab:feps}

\end{table*}



Since our work is motivated the lack of security definitions for existing Fully Encrypted Protocols, we identified a set of protocols that make an attempt to appear fully encrypted above the transport layer for analysis. Our goal in analyzing these protocols is to determine if these protocols satisfy our definitions, and to what extent they have identifying features that our security definitions are designed to address. We excluded closed-source protocols but note that many existing implementations of FEPs that do not appear in our list nonetheless adapt the approach of one of the protocols we analyze~\cite{outlinevpn,psiphon,tunnelblick}. We did not perform a thorough audit of the security engineering of the protocols' implementations, which we consider outside the scope of our paper.

\subsection{Methodology}

In each case we analyzed the protocol source code and documentation (if available) in order to identify the protocol behavior. To determine whether the protocols satisfy passive FEP security, we examined if the protocol outputs were all either random bytes or pseudorandom ciphertext. For active FEP security, we identified whether these fragments or messages were also authenticated, and for datastream protocols, whether the close behavior of the protocol satisfied our definition of a secure close function.

Protocol source code and documentation also informed our analysis of techniques for length obfuscation and of the minimum-length messages of the protocols. We examined these protocol aspects to understand how well \shaping{} was satisfied. Existing FEPs typically accomplish length obfuscation with padding, where extra bytes are added to messages. Padding can only increase the length of the output for a given input, and thus we expect to find that existing protocols possess minimum message lengths that can serve an undesirable identifying feature of the protocol. We determined whether padding was included in the protocols and how it was added by the code, and we confirmed the presence of padding in the experiments described below. In addition, we identified protocol message formats and layouts from these sources, which we used to form hypotheses on the minimum message lengths for each protocol, which we verified via experimentation. We note that the minimum lengths we determine technically apply to the amount of data written to the socket buffer by a \tf{Send} call, and the network stack may further fragment the message. However, such fragmentation is not common, and the amount written to the buffer is typically the amount sent in a network packet (\eg, we observed no fragmentation in our experiments).

We observed that channel-close behavior was rarely documented and in practice varied significantly between implementations, and so we relied almost entirely on our experiments to determine when protocols terminated a connection. We note that many protocols also include time-based close behavior, which we did not analyze. 

Open-source FEPs are designed to be flexible, and so all protocols in our list can be run under many configurations. Our primary goal was to identify the intended behavior of the protocol designers under recommended or default settings. Therefore, in our experiments we selected defaults, used recommended settings, and followed official examples to configure each protocol. 

Many assessed protocols are designed around a request-response architecture, and thus their message format, cryptographic structure, and general behavior may differ depending on whether the messages are outgoing ($\text{Client} \to \text{Server}$) or incoming ($\text{Server} \to \text{Client}$). Our protocol framework is unidirectional, and so we model each of these directions as a distinct protocol, combining them only when our results for both directions are identical. 

\subsection{Results}
Our results appear in Table \ref{tab:feps}. For padding approaches, we include only padding for the purpose of countering traffic analysis (\eg, not padding required to align a plaintext message to a block boundary before encryption), and padding that could apply to all messages (\eg, not padding only in the initial handshake phase).  We highlight the following results:
\begin{enumerate}
  \item Nearly all existing FEPs satisfy passive FEP security, but no datastream FEP satisfies active FEP security (Shadowsocks does in one direction). Also, the datastream FEPs are identifiable to an active attacker based on their close behavior. 
  \item No existing protocol, whether datastream or datagram, satisfies \shaping, and all protocols are identifiable by their minimum message length.
\end{enumerate}

We summarize our results in general and highlight specific results of interest below, providing the detailed descriptions of our methodology, analysis, and experiments in
\iffull
  Appendix \ref{appendix:experiments}.
\fi

We first describe the eight FEP implementations studied: (1) Shadowsocks~\cite{shadowsocks} is a fully encrypted SOCKS5 proxy with many implementations and the ability to proxy both datastream and datagram traffic. We examined the Shadowsocks-libev~\cite{shadowsockslibev} implementation in both TCP and UDP configurations. (2) We also studied a Shadowsocks implementation in the censorship-circumvention software suite V2Ray~\cite{v2ray}. (3) VMess~\cite{vmess} is a custom FEP also implemented inside V2Ray for censorship circumvention. (4) Obfs4~\cite{obfs4} is a FEP designed for Tor~\cite{dingledine2004tor}. (5) OpenVPN~\cite{openvpn} is open-source VPN software that functions in datastream or datagram mode, and there is a well-known patch available designed to obfuscate its traffic called the XOR patch~\cite{xor}. (6) Obfuscated SSH is modification of OpenSSH that fully encrypts its handshake messages, and it can be configured with a Pre-Shared Key. (7) kcptun is a FEP designed to be reliable and fast over very noisy networks. (8) Wireguard-SWGP~\cite{swgp} is a proxy for Wireguard~\cite{donenfeld2017wireguard} that fully encrypts its traffic. 

Datastream protocols largely terminate the connection immediately when an authentication tag failed to validate (``Auth Fail'' in Table \ref{tab:feps}). This behavior cannot be realized by any secure close function when the protocol ciphertext lengths are hidden (\eg, with a variable-length ciphertext and an encrypted length field, as is the case for Shadowsocks, VMess, Obfs4, and Obfuscated OpenSSH) because this means the close function, without access to encryption keys, must be able to identify a ciphertext boundary within the concatenation of all ciphertexts fragments from the sender. 

V2Ray protocols additionally include a ``Drain'' initial behavior (the protocol changes this behavior later in the connection) which upon decryption error delays terminating the connection until a pre-determined number of total bytes have been received, randomized based on a user and per connection. The drain approach does not satisfy our notion of a secure close function, as it does not apply if the pre-determined amount has been exceeded before the error. Thus, the V2Ray protocols cannot satisfy FEP-CCFA.

Most datastream protocols satisfy FEP-CPFA. Only Obfuscated OpenSSH, which transmits symmetric key material in the clear, and the OpenVPN XOR patch, which doesn't satisfy confidentiality~\cite{xue2022openvpn}, do not. On the other hand, only Shadowsocks-libev in one direction has a close behavior compatible with FEP-CCFA: the server holds the connection open indefinitely upon error as an intentional design feature. It does report errors to the application layer, though, technically violating the FEP-CCFA definition. However, this behavior could plausibly satisfy FEP-CCFA with a reasonable and minor modification to the notion (allowing errors) or to the transport protocol (suppressing errors).

Most datagram protocols satisfy FEP-CCA, since there are no closures in this setting. Only OpenVPN-XOR does not, and it also violates FEP-CPA in this setting as well.

Few protocols include padding of any sort (only Obfs4, VMess, and Wireguard-SWGP). The Shadowsocks datastream protocol had the same minimum message length in both implementations. However, other protocols varied significantly, with several (kcptun, OpenVPN-XOR, VMess, and Wireguard-SWGP) sending minimum-length messages as keepalives, which are thus transmitted whenever the application is quiescent. Such keepalives represent a novel identifying feature for FEPs, and they highlight the value of \shaping{}. Datastream FEPs typically had higher minimum message lengths because they all included a per-message nonce.

Three determinations of the minimum length require some care: Obfs4 servers upon initialization select a random set of random possible output lengths, which persists. Thus, each Obfs4 server \emph{has its own minimum message length}, with 44 bytes as the absolute minimum across all such setups. Shadowsocks-libev in the datagram setting can send messages of length 52 if directed to produce messages with invalid addresses\iffull(see Appendix \ref{appendix:experiments})\fi.
Wireguard-SWGP has a minimum length that depends on the configured MTU of the protocol, which the documentation suggests be set carefully, but the minimum message length is 75 with the default MTU of 1500.  

We also experimentally observed that the protocols in the V2Ray framework fail to satisfy datastream integrity (and thus cannot satisfy FEP-CCFA), dropping isolated messages silently when certain ciphertext bytes were modified in transit. We reported this as a security vulnerability to the project maintainers on January 4, 2024 and the issue was resolved in version 5.14.1.\footnote{The bug was caused by a failure to propogate errors properly when decryption fails during the data transport phase.} 

\section{Discussion}
\paragraph{Modeling limitations}
Our modeling choices limit our results in some ways. We follow existing models for stateful
communication channels, which, for simplicity, provide unidirectional rather than bidirectional
communication and omit time. Bidirectional communication can be achieved within our model by
composing two unidirectional channels, one in each direction. However, such channels are not allowed
to share state, including keys and other initialization parameters, and not observing this
restriction can violate the security guarantees. Such a limitation reduces potential efficiency over
the truly bidirectional setting. Similarly, omitting time from our model precludes some convenient
protocol features, such as timeouts leading to channel closures and time-based replay protections.
This limits both the design of new protocols and the analysis of existing ones. However, security
models have been designed that include time~\cite{backes2014tuc,canetti2017universally}, and our
models and constructions can be extended to work within such a setting.

\paragraph{FEP indistinguishability}
The ability of an adversary to distinguish FEPs from one another (as opposed to distinguishing them from non-FEPs) is an important security concern. Identifying the use of a specific FEP, implementation, or software version enables an adversary to deploy and tailor exploits. It can also enable user profiling, which may be particularly sensitive in the context of censorship circumvention, where software distribution may be performed via physical or social networks, and so the use of a specific implementation or software version can imply membership in such a network. 

Unfortunately, as we have shown in Section \ref{sec:analysis}, existing FEPs and their implementations can be effectively distinguished in practice because of a failure to satisfy one or more of the security definitions we present in this work. We consider our definitions a major step towards articulating and achieving this goal for FEPs, but using them to guarantee FEP indistinguishability requires some further determinations to be made that we consider outside the scope of our paper. In particular, a secure close function must be selected and a \shaping schedule or distribution must be determined. Further, some variability in these selections may be necessary among FEPs for practical functionality and efficiency reasons. 

\paragraph{Secure close functions}
Secure close functions can present practical concerns since they constrain protocol behavior. For example, they may require channels to leave connections open indefinitely on error (\eg, $\mathcal{C} \equiv 0$ as in Shadowsocks and our construction), deploy a less-efficient message format (\eg, the fixed-length ciphertexts of InterMAC~\cite{boldyreva2012security}), or terminate connections even if errors are not introduced into the ciphertext stream (\eg, if connections are closed after a fixed number of bytes or after a fixed byte sequence appears in the random ciphertext stream). However, some developers have already chosen in practice to make these tradeoffs in favor of security, both in Shadowsocks~\cite{shadowsockslibev}, which is the most widely deployed FEP in use, and libInterMAC~\cite{albrecht2019libintermac}, which has been thoroughly evaluated for efficiency. Additionally, the secure use of timeouts (which are outside our current model) to close connections may alleviate these concerns in some cases.

\paragraph{\shaping{}}
The \shaping property presents a tradeoff between security against traffic analysis and efficiency, to be set as appropriate for a given application. A thorough analysis of how best to use \shaping in particular use cases is important but requires the use of very different tools than we apply in this work. One simple option would be to send fixed-size messages at a constant rate, requiring predictable bandwidth and providing a maximum message latency. Alternately, one could sample message times and sizes from a known distribution for a target application (the transported application to maximize efficiency, or an uncensored application in the context of censorship circumvention). The main goal of our \shaping definition is to formalize a maximally flexible functionality \emph{at the transport layer}, providing a powerful interface that can be used to make traffic from any two distinct applications appear as similar as efficiency goals allow. 

\paragraph{Recommendations for FEP designers}
We intend our results to inform the design of existing and future FEPs. Many of the existing FEPs
could satisfy our novel security definitions with straightforward changes to the message
formats and protocol implementation. We make the following specific recommendations regarding the
two novel protocol features we introduce, secure close functions and \shaping{}:
\begin{enumerate}[leftmargin=*]
\item We suggest ideally adopting the close function of InterMAC, where the receiver closes after
error only at a multiple of $n$ bytes, for constant $n$, by adding a MAC at those positions. These
MACs would be in addition to any MACs needed to support variable-length records and could occur at
relatively low frequency, where such frequency should be chosen to be acceptable across applications
and protocols to maximize FEP indistinguishability (\eg{} every 10K bytes). We further suggest
closing the connection after a certain amount of time has passed with no ciphertext received, which
is technically outside our model. The combination of these close behaviors would permit connections
to eventually close after an error is observed without adding too much bandwidth overhead. An
alternative suggestion is simply to adopt a connection timeout after not receiving any ciphertext.
This behavior does permit an active adversary to keep a connection open by continually sending
traffic, but it is a more minor change to implement.
\item To implement \shaping, datastream FEPs should be extended to buffer content that does not fit
within the prescribed output length $p$ and to produce padding when $p$ exceeds the amount necessary
for the message. Datagram FEPs should be extended to include the use of null messages and to refuse
to send messages that cannot fit within the prescribed length. FEP developers should thus implement
\shaping, although existing applications need not necessarily use it and instead await more work on
appropriate traffic patterns.
\end{enumerate}

\section{Related Work}
We build on the model of Fischlin~\etal~\cite{DIAS} to study data streams in the presence of fragmentation, a problem also been studied elsewhere~\cite{boldyreva2012security,albrecht2009plaintext,albrecht2019libintermac,albrecht2016surfeit}. Fragmentation is a problem for FEPs, as plaintext length fields cannot be used. The notion of Boundary Hiding~\cite{boldyreva2012security} is related to but does not imply our FEP notions because, for one, it does not enforce random-appearing outputs.

Some aspects of Fully Encrypted Protocols have been studied, such as active probing~\cite{frolov2020detecting,beznazwy2020china}, and attacks on confidentiality~\cite{ji2022security,shadowsocksredirect}. Much research has focused on detecting FEPs~\cite{wang2015seeing,wu2023great,frolov2019use,xue2022openvpn} using a variety of approaches. Fifield~\cite{fifield2023flaws} outlined a series of implementation weaknesses in existing FEPs. Bernstein~\etal~\cite{bernstein2013elligator} developed a method to encode elliptic curve points as uniformly random strings. 

On channel closures, Boyd and Hale~\cite{boyd2017secure} consider channels with in-band intentional termination signals that generate channel closures, and Marson and Poettering~\cite{marson2017security} provide security definitions for fully bidirectional channels. Hansen~\cite{hansen2020cryptographic} and Albrecht~\etal~\cite{albrecht2019libintermac} discuss similar issues to the channel closures we discuss in our work around the difficulty of realizing active boundary hiding, where applications themselves leak information to an adversary through their behavior. 

Fischlin~\etal~\cite{robust} analyze dTLS~\cite{rescorla2022rfc} and QUIC~\cite{langley2017quic} in the datagram setting, considering especially robustness to message drops. Our datagram model can also apply to other encrypted transport protocols like IPSec~\cite{doraswamy2003ipsec}. Stateful
encryption models similar to our datagram model are given by Bellare et al.~\cite{bellare2002authenticated} and Kohno et al.~\cite{cryptoeprint:2003/177}.

Some FEPs designs have come from elsewhere than the open-source community. Obfs4~\cite{obfs4} is based on Scramblesuit~\cite{winter2013scramblesuit} and has been forked under the new name Lyrebird~\cite{lyrebird}. InterMAC, which is a near-FEP, has been both analyzed~\cite{boldyreva2012security} and implemented~\cite{albrecht2019libintermac}. The IETF proposal for a Pseudorandom Extension of cTLS~\cite{ctls-pseudorandom} suggest a FEP encoding of TLS to improve protocol security and privacy rather than for censorship circumvention.

\section{Conclusion and Future Work}
A natural enhancement to our constructions would be a fully encrypted key exchange for forward
secrecy. Obfs4~\cite{obfs4} uses Elligator~\cite{bernstein2013elligator} for this purpose, but it
remains to prove its security as well as investigate other techniques. Relatedly, one could
formulate a notion of forward metadata secrecy, where the FEP properties may still be preserved even
if a long-term key is compromised. In addition, FEP notions should be explored in models where
securely creating shared state is not trivial but bidirectional communication is possible.

There are also other practical concerns to consider. Active probing attacks are a common problem for real-world FEPs, which we do not address in this work. We also do not optimize efficiency in our FEP protocol constructions. Finally, while we introduce the \shaping capability, we do not answer the distinct and important question of how traffic \emph{should be shaped}. 

\section*{Acknowledgments}

This work was supported by the Office of Naval Research
(ONR) and the Defense Advanced Research Projects Agency
(DARPA). We would also like to thank Ben Robinson, who 
helped perform the Wireguard experiments, and David Fifield and the
anonymous referees who reviewed this paper for their 
helpful suggestions and advice.

\bibliographystyle{style/ACM/ACM-Reference-Format}
\balance
\bibliography{references}
\iffull
        \appendix

\section{Standard Datastream Security Definitions} \label{app:dstream-sec}
We extend the datastream confidentiality definition from Fischlin et al.~\cite{DIAS} to include
channel closures:
\begin{definition} \label{def:ind-cpfa-cl}
    A channel satisfies $\textsf{IND-CPFA-CL}$ if, for any PPT adversary $\mathcal{A}$ and security parameter $\lambda$, $\left|P\left[\mathsf{Exp}^{\textsf{IND-CPFA-CL},b}_{\mathcal{A}}(1^{\lambda}) = 1 | b \overset{{\scriptscriptstyle \operatorname{R}}}{\leftarrow} \{0,1\}\right] - 1/2\right|$ is negligible in $\lambda$.
\end{definition}

In the security experiment for IND-CPFA-CL (Algorithm~\ref{alg:obfs-cpfa-cl-exp}), the adversary is
given a sending oracle (Algorithm~\ref{alg:CPFA-LOR}) and a receiving oracle
(Algorithm~\ref{alg:CPFA-RECV}). Note that the receiving oracle only returns the channnel close
flag, as only it would be observable to a passive adversary.

\begin{algorithm}[H]
    \caption{$\mathsf{Exp}^{\textsf{IND-CPFA-CL},b}_{\mathcal{A}}(1^{\lambda})$} \label{alg:obfs-cpfa-cl-exp}

    \begin{algorithmic}[1]
        \State $(\tok[S]{st}, \tok[R]{st}) \gets \tf{Init}(1^\lambda)$
        \State $C_S \gets \emptyList$
        \State $C_R \gets \emptyList$
        \State $b' \gets \mathcal{A}^{\mathcal{O}^b_{\tf{LoR-CPFA-CL}}(), \mathcal{O}^b_{\tf{Recv-CPFA-CL}}()}(1^\lambda)$

        \State \Return $b' = b$
    \end{algorithmic}
\end{algorithm}

\noindent\begin{minipage}[t]{.5\textwidth}

\begin{algorithm}[H]
    \caption{\label{alg:CPFA-LOR}$\mathcal{O}^b_{\tf{LoR-CPFA-CL}}(m_0,m_1,p,f)$}
    \begin{algorithmic}[1]
            \If{$|m_0| \neq |m_1|$}
                \State \textbf{return} $\bot$
            \EndIf
        \State $(\tok[S]{st}, c) \gets  \tf{Send}(\tok[S]{st},m_b,p,f)$
        \State $C_S \getsAppend{C_S}{c}$
        \State \Return $c$
    \end{algorithmic}
        \end{algorithm}
    \end{minipage}
\hfill
    \begin{minipage}[t]{.5\textwidth}

    \begin{algorithm}[H]

        \begin{algorithmic}[1]
            \If{$\left( \cat C_R\right) \Vert c \npreceq \cat C_S$}
                \State \textbf{return} $\bot$
            \EndIf
            
        \State $(\tok[R]{st}, m, \tok{cl}) \gets  \tf{Recv}(\tok[R]{st},c)$
        \State $C_R \getsAppend{C_R}{c}$
        \State \Return $(\emptySt, \tok{cl})$
    \end{algorithmic}
            \caption{\label{alg:CPFA-RECV}$\mathcal{O}^b_{\tf{Recv-CPFA-CL}}(c)$}

    \end{algorithm}

        \end{minipage}\\




\section{Standard Datagram Security Definitions} \label{app:dgram-sec}
We adapt standard security definitions in the atomic-message context to our stateful channel model.
These definitions provide security with respect to a probabilistic polynomial-time (PPT) adversary.
Note that they use the datagram channel functions $\tf{Init}$, $\tf{Send}$, and $\tf{Recv}$.

For confidentiality, we use the \indcpa{} and \indcca{} definitions:
\begin{definition} \label{def:ind-cxa-security}
    A channel satisfies \indcxa{}, $x\in \{\textrm{CPA, CCA}\}$ if, for a security parameter $\lambda$ and PPT adversary
    $\mathcal{A}$, $\left|P\left[\mathsf{Exp}^{\textsf{\indcxa},b}_{\mathcal{A}}(1^{\lambda}) =
    1 \big| b \overset{{\scriptscriptstyle \operatorname{R}}}{\leftarrow} \{0,1\}\right] - 1/2\right|$ is negligible in $\lambda$.
\end{definition}
In the related security experiment (Algorithm~\ref{alg:ind-cxa-exp}), the adversary is given
adaptive access to a sending oracle (typically called a left-or-right oracle,
Algorithm~\ref{alg:ind-cxa-lor}). With \indcca{} security, it also gets adaptive access to a
receiving oracle (Algorithm~\ref{alg:ind-cca-recv}).

\begin{algorithm}
    \caption{$\mathsf{Exp}^{\textsf{\indcxa}, b }_{\mathcal{A}}(1^{\lambda})$} \label{alg:ind-cxa-exp}
    \begin{algorithmic}[1]
        \State $(\tok[S]{st}, \tok[R]{st}) \gets \tf{Init}(1^\lambda)$
        \State $C \gets \emptyset$
        \State $b' \gets \left\{
            \begin{array}{ll}
                \mathcal{A}^{\mathcal{O}^b_{\tf{\indlor}}()}(1^\lambda) & \textrm{if $x = $ CPA}\\
                \mathcal{A}^{\mathcal{O}^b_{\tf{\indlor}}(), \mathcal{O}_{\tf{\indccarecv}}()}(1^\lambda) & \textrm{if $x =$ CCA}
            \end{array}\right.$
        \State \Return $b' = b$
    \end{algorithmic}
\end{algorithm}

\begin{algorithm}
    \caption{$\mathcal{O}^b_{\tf{\indlor{}}}(m_0,m_1,p)$} \label{alg:ind-cxa-lor}
    \begin{algorithmic}[1]
    		\If{$\nullsym\in \{m_0, m_1\} \land m_0 \neq m_1$}
    			\State \textbf{return} $\bot$
            \ElsIf{$|m_0| \neq |m_1|$}
                \State \textbf{return} $\bot$
    		\EndIf

        \State $(\tok[S]{st}, c) \gets  \tf{Send}(\tok[S]{st},m_b,p)$
        \If{$c \neq \bot$}
        \State $C \gets C \cup \{c\}$
        \EndIf

        \State \Return $c$
    \end{algorithmic}
\end{algorithm}

\begin{algorithm}
    \caption{$\mathcal{O}_{\tf{\indccarecv}}(c)$} \label{alg:ind-cca-recv}
    \begin{algorithmic}[1]
        \State $(\tok[R]{st}, m) \gets \tf{Recv}(\tok[R]{st}, c)$
        \If{$c \notin C \land m \neq \bot \land m \neq \nullsym$}
            \State \textbf{return} $m$
        \Else
            \State \textbf{return} $\bot$
        \EndIf
    \end{algorithmic}
\end{algorithm}

For integrity, we use the \intctxt{} definition:
\begin{definition}\label{def:int-ctxt-security}
    A channel satisfies \intctxt{} if, for a security parameter $\lambda$ and PPT adversary
    $\mathcal{A}$, $\left|P\left[\mathsf{Exp}^{\textsf{\intctxt}}_{\mathcal{A}}(1^{\lambda}) =
    1\right] \right|$ is negligible in $\lambda$.
\end{definition}
In the related security experiment~\ref{alg:int-ctxt}, the adversary gets access to a sending
oracle (Algorithm~\ref{alg:int-ctxt-send}) and a receiving oracle
(Algorithm~\ref{alg:int-ctxt-recv}).

\begin{algorithm}
    \caption{$\mathsf{Exp}^{\textsf{\intctxt}}_{\mathcal{A}}(1^{\lambda})$} \label{alg:int-ctxt}
    \begin{algorithmic}[1]
        \State $(\tok[S]{st}, \tok[R]{st}) \gets \tf{Init}(1^\lambda)$
        \State $C \gets \emptyset$
        \State $b \gets 0$
        \State $\mathcal{A}^{\mathcal{O}_{\tf{\intctxtsend}}(), \mathcal{O}_{\tf{\intctxtrecv}}()}(1^\lambda)$
        \State \Return $b$
    \end{algorithmic}
\end{algorithm}

\begin{algorithm}
    \caption{$\mathcal{O}_{\tf{\intctxtsend}}(m,p)$} \label{alg:int-ctxt-send}
    \begin{algorithmic}[1]
        \State $(\tok[S]{st}, c) \gets  \tf{Send}(\tok[S]{st},m,p)$
        \If{$c \neq \bot$}
        \State $C \gets C \cup \{c\}$
        \EndIf
        \State \Return $c$
    \end{algorithmic}
\end{algorithm}

\begin{algorithm}
    \caption{$\mathcal{O}_{\tf{\intctxtrecv}}(c)$} \label{alg:int-ctxt-recv}
    \begin{algorithmic}[1]
        \State $(\tok[R]{st}, m) \gets \tf{Recv}(\tok[R]{st}, c)$
        \If{$c \notin C \land m \neq \bot \land m \neq \nullsym$}
            \State $b \gets 1$
        \EndIf
        \State \textbf{return} $m$
    \end{algorithmic}
\end{algorithm}

\section{Proofs for Datastream Relations Between Notions}
\label{appendix:stream-relations}

We provide proofs for each theorem from Section \ref{sec:stream-relations}.\\


\noindent First, we prove Theorem \ref{prop:lenreg-fep-ind}. 

\begin{proof}
    We adapt the idea for the straightforward reduction from IND\$-CPA to IND-CPA.  Consider the contrapositive, and suppose that there is a CPFA adversary $A_{\text{CPFA}}$ who can win the IND-CPFA game with non-negligible advantage $\epsilon$, and that the channel is length regular. Then we define an adversary for the \passiveObfs game $A_O$. $A_O$ plays the \passiveObfs game and runs a copy of $A_{\text{CPFA}}$. $A_O$ picks a bit at random $b_0$ and plays as the challenger for $A_{\text{CPFA}}$ in the IND-CPFA game. Every time $A_{\text{CPFA}}$ calls $\mathcal{O}_{LOR}$ with input $(m_1, m_0, f)$,  $A_O$ sends $m_{b_0}$, $p=-1$, $f$ to $\mathcal{O}_\textsc{Send}$, and passes the oracle's return value on to $A_{\text{CPFA}}$. Finally, when $A_{\text{CPFA}}$ produces an output bit $b_1$, $A_O$ outputs $1$ iff $b_0 = b_1$.

We calculate $Adv(A_O) = |P(A_O \to 1 | b=1) - P(A_O \to 1 | b=0)|$.

First, we note that if $b=1$, $A_O$ is interacting with the \passiveObfs challenger, and so $A_{\text{CPFA}}$ receives uniformly random strings of the content-independent length, carrying no information about $b_0$. Thus, $P(A_O \to 1 | b=1)$ is the advantage of $A_{\text{CPFA}}$ when interacting with a LoR oracle that produces uniformly random bytes, and is $0$. 

Second, if $b=0$, $A_{\text{CPFA}}$ receives a stream of genuine ciphertexts, so the second term reduces to the probability that $A_{\text{CPFA}}$ correctly guesses $b_0$ in the IND-CPFA game, with $b_0$ selected at random. 
Thus $Adv(A_O) = |0 - P(A_{\text{CPFA}} \textrm{wins the IND-CPFA game})|$ which completes the reduction.

\end{proof}
We prove Theorem \ref{prop:corr-cpfa-cl}. 

\begin{proof}
Since the close function is secure, it can be evaluated by an adversary for the IND-CPFA game, which allows for a trivial reduction. 
\end{proof}

We prove Theorem \ref{prop:cst-ccfa}. 

\begin{proof}
Suppose there is an adversary $A_{\tok{INT}}$ that wins the INT-CST experiment with non-negligible probability $\epsilon$. We construct an adversary $A$ that has non-negligible advantage in the \activeObfs experiment. 

First we define $A$, an adversary for the \activeObfs experiment which runs a copy of $A_{\tok{INT}}$ and acts as its challenger in the INT-CST experiment, passing on inputs to \tf{Send}, \tf{Recv} from $A_{\tok{INT}}$ onto its own oracles \tf{Send}, \tf{Recv}. $A$ returns to its challenger the value of $\tok{win}$ at the end of the INT-CST experiment. 

We observe that if $b=1$, $A$ always outputs $0$. If $b=0$, the behavior of the oracles for \tf{Send} and \tf{Recv} in the \activeObfs experiment is, by construction, defined to be identical to the behavior of these oracles in the INT-CST experiment. Thus, $A_{\tok{INT}}$ wins (and therefore $A$ outputs $1$) with probability $\epsilon$. Therefore the advantage of $A$ is $\epsilon/2$, which is non-negligible whenever $\epsilon$ is. 
\end{proof}

We prove Theorem \ref{thm:new-general-comp}. 
\begin{proof}
	Consider an adversary $A$ with non-negligible advantage in the FEP-CCFA-$\mathscr{C}$ game. We make a series of transitions in order to show that the advantage of $A$ is equal to an advantage of some adversary in the FEP-CPFA game. First, by correctness the close output of \tf{Recv} is identically distributed to the output of $\mathscr{C}$, and so we transition to a hybrid execution $E_1$ where $A$ receives the close output of $\mathscr{C}$ rather than the genuine output of \tf{Recv}, which is just a view-change. Then we observe that ERR-FREE gives that the channel produces no errors, and so by INT-CST, the probability that $A$ produces some nonempty output from $\mathcal{O}^b_{\tf{Recv}}(c)$ in the case where $b=0$ is some negligible function $\epsilon_{\tok{INT}}$ (if $b=1$, the plaintext output of $\mathcal{O}^b_{\tf{Recv}}(c)$ is identically $\emptySt$ by construction). Then we transition to a hybrid execution $E_2$ where $\mathcal{O}^b_{\tf{Recv}}(c)$ is replaced with a function that calls $\mathscr{C}$ and responds with $\emptySt$ and its output. Finally, we define a new adversary $A_1$ for the FEP-CPFA game, which runs $A$ in $E_2$ as follows, initializing three lists $C_S$, $C_R$, $C_{\tok{cl}}$.  Every time $A$ queries its \tf{Send} oracle, $A_1$ passes this on to its \tf{Send} oracle and records the output, concatenating it with $C_S$. Every time $A$ queries \tf{Recv}, $A_1$ appends the ciphertext fragment to $C_R$ and calls $\tok{cl} \gets \mathscr{C}(C_S, C_R, C_{\tok{cl}}, c)$, and responds to $A$ with $(\emptySt,\tok{cl})$, appending $\tok{cl}$ to $C_{\tok{cl}}$. Since $\mathscr{C}$ is secure, $A_1$'s simulation of it is perfect and this experiment is exactly $E_2$, so $A_1$ wins the \passiveObfs game whenever $A$ wins $E_2$, completing the reduction.

\end{proof}

\section{FEP-CPFA Does Not Imply IND-CPA}
\label{appendix:noCPFA}
Assume the AEAD scheme from Section \ref{sec:aead-scheme}. Since the scheme is length-regular, the encryption of $1$ byte always produces a ciphertext of a fixed length $1 + \tagSize$, which we denote $\ell$.

\begin{enumerate}
    \item $\tf{Init}$ returns a symmetric key for our AEAD scheme defined in Section \ref{sec:aead-scheme}. 
    \item $\tf{Send}(\tok[S]{st}, m)$ parses the string of bytes $m$ as a binary encoding of a natural number $n_m$ and outputs $\tf{Rand}(n_m) \Vert \tf{Enc}(0)$ where $\tf{Rand}$ is a function that produces uniformly random byte strings of the given length. Since $\tf{Enc}$ is length regular, this output ciphertext always has length $n_m + \ell$.
    \item $\tf{Recv}(\tok[R]{st},c)$ buffers the bytes in $c$ exactly as the datastream channel in \cite{DIAS}, and applies $\tf{Dec}$ sequentially to every subsequence of bytes of length $\ell$ in its buffer. If decryption succeeds and the message text is $0$, all bytes in the buffer of $\tf{Recv}$ that precede this valid ciphertext are counted producing a natural number $n_m$, which is treated as the encoding of a byte string $m$. All bytes in the buffer up to and including this valid ciphertext are removed, and $\tf{Recv}$ outputs $m$. Otherwise, $\tf{Recv}$ outputs nothing.
\end{enumerate}

\begin{theorem} The channel defined above satisfies \passiveObfs and correctness but does not satisfy IND-CPFA.\end{theorem}
\begin{proof}
First, it is clear that the channel does not satisfy IND-CPFA. We note that the
$\tf{Send}$ algorithm has no buffering behavior, and immediately leaks the
plaintext content via the ciphertext length. An adversary can simply submit one
of two bytes $m_0 = 1$, and $m_1 = 255$. Since these plaintexts have the same
length, the left-or-right oracle will produce a ciphertext using $\tf{Send}$ of
length either $\ell + 1$ or $\ell + 255$, allowing the adversary to determine
the oracle's secret hidden bit with advantage $1$.

Second, we consider correctness. We note that if a full ciphertext is received
by $\tf{Recv}$, and no bytes are added in the receiver stream $C_R$, then the
number of bytes between valid ciphertexts is exactly $|m|$ for each valid
message, meaning that $\tf{Recv}$ produces the correct plaintext result once it
receives the entire corresponding ciphertext.

Finally we turn to the \passiveObfs property. Consider an adversary $A_O$ in the
\passiveObfs game. We construct an adversary $A$ that has the same advantage in
the IND\$-CPA game for the underlying AEAD scheme. $A$ runs a copy of $A_O$ and
each time $A_O$ submits a message $m$, $A$ submits $0^{\inLength}$ to the
IND\$-CPA challenger, receiving a ciphertext $c$. $A$, acting as a challenger
for $A_O$, produces $\tf{Rand}(|m|) \Vert c$ and returns this to $A_O$. Then, if
$A$ is producing random results, $A_O$ is playing the \passiveObfs game against
an adversary with $b=1$, and otherwise, $A_O$ is playing the \passiveObfs game
with $b=0$. Thus, the advantage of $A_O$ in the \passiveObfs game and the
advantage of $A$ in the IND\$-CPA game are the same.
\end{proof}

\section{Datastream Construction Proofs}
\label{appendix:datastream-construction-proofs}
We prove Theorem \ref{thm:datastream-construction-all} by establishing each property individually in the results that follow. 

First we prove correctness, \shaping, and channel length regularity for our construction. 

\begin{theorem}
\label{thm:stream-correctness}
The channel construction in Figure \ref{fig:stream-construction} satisfies correctness for secure close function $\mathscr{C} \equiv 0$.
\end{theorem}

\begin{proof}
\textbf{Termination}. First we argue that \tf{Send} and \tf{Recv} always terminate. For \tf{Send}, consider the conditional on line $3$: \tf{Send} only fails to terminate when the conditions here are not met, and if this condition fails, the subsequent recursive call to \tf{Send} occurs only after progress has been made towards the condition. We note $\tok{buf} = \emptySt$ except for the first call to \tf{Send}, and consider $|\tok{obuf}|$, noting that it is incremented with each call by at least $2*\lenlen + 2> 0$ in Line $13$, $17$, giving the desired result. For \tf{Recv}, we examine the while loop, and observe that each conditional branch either returns, breaks, or reduces $|\tok{buf}|$ by at least $\lenlen > 0$, progressing towards the termination condition. 

\textbf{Stream Preservation}. We assume all preconditions for the correctness requirement, and consider an arbitrary sequence of \tf{Send} calls interleaved in some order with an arbitrary sequence of \tf{Recv} calls. We label the concatenation of the plaintext fragments that are ever passed to \tf{Send} calls as $M_0$, and consider the following buffers throughout the execution of this sequence of function calls:
\begin{itemize}
\item $M_s$, the concatenation of all plaintext fragments not yet passed to \tf{Send}
\item $\tok[S]{buf}$, the internal buffer $\tok{buf}$ within the state of \tf{Send}
\item $\tok{obuf}$
\item $C$, the concatenation of ciphertext fragments output by \tf{Send} but not yet passed to \tf{Recv}
\item $\tok[R]{buf}$, the internal buffer within the state of \tf{Recv}
\item $M_r$, the concatenation of the outputs of \tf{Recv}
\end{itemize}

\begin{lemma} \label{lemma:stream-construction-correctness-buffers} At every point in the the execution, the concatenation of all ciphertext buffers $(\tok[R]{buf}\Vert C \Vert \tok{obuf})$ can be parsed as a (possibly empty) sequence of complete, correct, length-block payload-block ciphertext pairs (with the payload blocks possibly equal to $\emptySt$).

\end{lemma}
\begin{proof}
We follow bytes as they pass through the three ciphertext buffers: first, every call of \tf{Send} either returns immediately or appends a length-payload block pair to $\tok{obuf}$. By the length regularity of the AEAD scheme, $\ell_c$ is the size of the ciphertext constructed in Line $17$, and the correctness of the padding clear by inspection. Then, every fragment $c$ removed from $\tok{obuf}$ is removed from the front of the buffer, and is returned by \tf{Send}, meaning that $C \Vert (c \Vert \tok{obuf}) = (C \Vert c) \Vert \tok{obuf}$ remains fixed. Similarly, bytes are removed from $C$ and passed into $\tok[R]{buf}$ by the same mechanism, fixing the concatenated ciphertext buffers. Finally, we show that \tf{Recv} only removes full length-payload block pairs from \tf{Recv}. But this is clear from the conditions in Lines 6, 7, 11: \tf{Recv} waits to receive a full length block, and removes only when the full following payload block has been received (in which case the payload block is removed as well).
\end{proof}

We then define a function $r$, which operates on ciphertext blocks by recording and tracing each block and its plaintext as it is constructed and appended to $\tok{obuf}$ within \tf{Send}. When applied to a sequence of ciphertext block pairs, $r$ reverses the encryption of each payload block, strips the padding length prefix byte and following padding bytes, and concatenates the resulting plaintext for its output, and we note that for two block pairs $b_1^0 \Vert b_1^1$, $b_2^0 \Vert b_2^1$, that $r(b_1^0 \Vert b_1^1 \Vert b_2^0 \Vert b_2^1) = r(b_1^0 \Vert b_1^1) \Vert r(b_2^0 \Vert b_2^1)$.

We define the value \[I = M_r \Vert r(\tok[R]{buf} \Vert C \Vert \tok{obuf}) \Vert \tok[S]{buf} \Vert M_s\] and define our \emph{buffer invariant} to be the property that $I = M_0$, under which it is immediately clear that $M_r \prec M_0$
and proceed by induction, showing $I$ remains unmodified throughout each execution, and  considering last the case where $f=1$ for the final call to \tf{Send}. 

First, as a base case, we observe that $I$ holds trivially before any \tf{Send} calls are executed, since in this case $M_0 = M_s$, and all other buffers are empty. 

Next consider a call $\tf{Send}(m,p,f)$. First, $m$ is concatenated to $\tok[S]{buf}$, which preserves our invariant since it is removed from the beginning of $M_s$ and placed in the end of $\tok[S]{buf}$, leaving $\tok[S]{buf} \Vert M_s$ fixed. If \tf{Send} returns via the conditional in Line 3, bytes are transferred from $\tok{obuf}$ to $C$, leaving $C \Vert \tok{obuf}$ fixed. Otherwise, a new fragment $c$ is produced and added to $\tok{obuf}$. We calculate the new value of $I$, labeled $I'$. $c$ consists of two ciphertext blocks, and $r(c) = \emptySt \Vert \tok[S]{buf}[1..o]$. Then we can write:
\[I' =  M_r \Vert r(\tok[R]{buf} \Vert C \Vert \tok{obuf})\Vert r(c) \Vert \tok[S]{buf}[o+1..] \Vert M_s\]
or
\[I' =  M_r \Vert r(\tok[R]{buf} \Vert C \Vert \tok{obuf})\Vert \tok[S]{buf}[1..o] \Vert \tok[S]{buf}[o+1..] \Vert M_s\]

which is equal to $I$. We observe that if \tf{Send} is called recursively, the arguments above apply again, except that every recursive call has $\emptySt$ as its plaintext fragment argument.

We have already shown \tf{Recv} removes block pairs from $\tok[R]{buf}$, so it remains to be shown that the decryptions succeed with in \tf{Recv} and that this process produces the same result on blocks as $r$. The nonce (sequence number) of every nonempty block matches the one assigned in its construction in \tf{Send}, and blocks are parsed as constructed, so decryption succeeds, and the padding is stripped by parsing the first payload block byte as $r$ does. 

\textbf{Flushing}. Finally, we show that if $f=1$ for the final \tf{Send} call, and that if all output bytes of \tf{Send} have been passed to \tf{Recv}, that $M_r = M_0 = I$, or equivalently, given our invariant property, that $M_S$, $\tok[S]{buf}$, and $r(\tok[R]{buf}\Vert C \Vert \tok{obuf})$ are all equal to $\emptySt$. $M_S$ is empty by assumption, as is $C$, since the final \tf{Send} call has been executed and all outputs of \tf{Send} have been passed to \tf{Recv}. Consider the final \tf{Send} call: first, we note \tf{Send} cannot terminate unless $\tok{buf} = \emptySt$ by the conditional in Line 3, and \tf{Send} always terminates. $\tok{obuf}$ is empty since in the last \tf{Send} execution, $f=1$ so in Lines 4-5, all of $\tok{obuf}$ is removed. Then since $C = \tok{obuf} = \emptySt$, it remains to show $r(\tok[R]{buf}) = \emptySt$. Since the other ciphertext buffers are empty and all ciphertext takes the form of correct block pairs by Lemma \ref{lemma:stream-construction-correctness-buffers}, we have that $\tok[R]{buf}$ is some sequence of valid block pairs. But \tf{Recv} has consumed all of the correct blocks inside of its buffer by construction, so $\tok[R]{buf}$ is empty. 

\textbf{Closures}. Close predictability and correctness are trivial, since our construction never closes the channel and therefore satisfies $\mathscr{C} \equiv 0$. 

 \end{proof}

Next, we show the channel satisfies channel length regularity and \passiveObfs, which are straightforward and follow primarily from the similar atomic properties for our AEAD encryption scheme. 

\begin{theorem} \label{prop:stream-construction-ch-reg} The channel construction in Figure \ref{fig:stream-construction} satisfies channel length regularity.
\end{theorem}

\begin{proof}
By assumption \tf{Enc} is length regular, and no logic in \tf{Send} depends on the content of any plaintext message or buffer, only their lengths. The result follows.  
\end{proof}

\begin{theorem} \label{prop:stream-construction-fep-cpfa} The channel construction in Figure~\ref{fig:stream-construction} satisfies \passiveObfs.
\end{theorem}

\begin{proof}
Straightforward, since all outputs of \tf{Send} are produced by \tf{Enc}, which satisfies IND\$-CPA by assumption. 
\end{proof}
    
\begin{theorem} \label{prop:stream-construction-integrity} The channel construction in Figure \ref{fig:stream-construction} satisfies INT-CST.
\end{theorem}

\begin{proof}
Since our construction is structurally similar to the AEAD construction in Fischlin \etal~\cite{DIAS}, we adopt the same strategy for our proof and we reduce the INT-CST game to the AUTH game for our underlying AEAD encryption scheme by assuming there is some adversary $A$ which wins the INT-CST game against our construction with advantage $\epsilon$, and constructing an adversary $B$ that runs $A$ and simulates the behavior of \tf{Send} and \tf{Recv}, replacing the encryption and decryption calls with their corresponding oracles in the AUTH experiment. For the details, we refer to Fischlin \etal~\cite{DIAS}, and we note where our proof differs. We note that the sequence of blocks recorded by $B$ are entirely ciphertext blocks, in contrast to the previous approach which interleaved plaintext length fields between each ciphertext block, however the sequence numbers are structured identically, incrementing by one with each block. Thus, $B$ records ciphertext blocks when they are passed to \tf{Send} and ``decrypts'' each block when a genuine ciphertext fragment is passed to \tf{Recv}. The argument is the same up until the first non-genuine ciphertext fragment $c'$ is passed to \tf{Recv}. $B$ extracts the genuine ciphertext blocks from the beginning of $c'$ and "decrypts" them. We note that for our channel, blocks can produce empty plaintext upon decryption and indeed that an AUTH forgery can take place within a ciphertext block that causes \tf{Recv} to produce nothing (\eg, if a length block is forged, or a payload block full of padding), but that violating INT-CST requires that $A$ produce inputs to \tf{Recv} that cause it to generate a \emph{nonempty} plaintext fragment. This must, by construction of \tf{Recv}, occur within a fragment of ciphertext which is parsed by \tf{Recv} as a payload block with non zero length, \ie in Line 14. 
This means that the plaintext fragment to be produced by \tf{Recv} as a result of deviating ciphertext, $m'$ must be the result of decrypting a payload block, and so if a deviation occurs within the preceeding length block, this is already a forgery and $B$ can wait to collect $\lenlen$ bytes to produce the forged length block; by assumption, if \tf{Recv} is to produce a ciphertext, \tf{Dec} will produce some valid nonnegative integer when applied to the block or \tf{Recv} sets $\tok{fail} = 1$, precluding future outputs. If the first deviation occurs within a payload block, the preceeding length block is genuine and has length $\ell$ and so $B$ can produce $\emptySt$ until $\ell$ bytes have been received, producing the full ciphertext block forgery. Thus, $B$ can construct the forgery for our protocol and INT-CST is satisfied. 

\end{proof}

\begin{theorem} \label{corr:stream-construction-ccfa} The channel construction in Figure \ref{fig:stream-construction} satisfies \activeObfs and IND-CCFA.
\end{theorem}

\begin{proof}
	We observe that since our construction never produces an error, it satisfies ERR-FREE. Then by Theorem \ref{thm:stream-correctness} (and noting the zero function is trivially a secure close function), Theorem \ref{prop:stream-construction-fep-cpfa}, and Theorem \ref{prop:stream-construction-integrity}, we apply Theorem \ref{thm:new-general-comp} and have that the construction satisfies FEP-CCFA. 
\end{proof}

\begin{theorem} \label{prop:stream-construction-length-shaping} The channel construction in Figure \ref{fig:stream-construction} satisfies \shaping.
\end{theorem}
\begin{proof}
First we observe that every recursive call of \tf{Send} passes on $f$ and $p$ unaltered, and that \tf{Send} only terminates on Line $6$, producing a value calculated in Line $4$. As argued in Thm.~\ref{thm:stream-correctness}, \tf{Send} always terminates. Therefore, it suffices to analyze the length of the value produced on Line $4$ with the values of $p$, $f$ provided to the initial call of Send. If $f=0$, $c$ is the first $p$ bytes of $\tok{obuf}$, which has at least $p$ bytes due to the conditional on Line $3$. If $f=1$, the length of $c$ is the larger of $p$ and $|\tok{obuf}|$ satisfying the requirement there as well. 
\end{proof}

\section{Proofs for Datagram Relations Between Notions}
\label{appendix:datagram-relations-proofs}
We prove Theorem~\ref{prop:datagram-cca-integrity}. 
\begin{proof}

Suppose a channel fails to satisfy INT-CTXT. Then there exists some PPT adversary for the INT-CTXT experiment $A$ that generates some non-suppressed, non-error, non-null output from  $\mathcal{O}_{\tf{\intctxtrecv}}$ with non-negligible probability. 
We then consider $A$'s behavior if instantiated as an adversary for the FEP-CCA experiment; if $b=1$ , the \tf{Recv} oracle always produces $\bot$ regardless of the behavior of $A$, but if $b=0$, the behavior of the \tf{Recv} oracle is identical to the \tf{Recv} oracle for the INT-CTXT game, thus by assumption $A$ produces some output $m \neq \bot$ from $\mathcal{O}^b_{\tf{\fepdgramrecv}}$ with non-negligible probability. Then it is clear that an adversary that runs $A$ in the FEP-CCA game and outputs $1$ exactly when non-error output is produced from $\mathcal{O}^b_{\tf{\fepdgramrecv}}$ has advantage equal to half of the advantage of $A$, which is also non-negligible.  
\end{proof}
We prove Theorem~\ref{prop:datagram-composition}. 

\begin{proof}
We sketch the proof as follows, adopting the same approach as Theorem \ref{prop:lenreg-fep-ind} and proof of the general composition theorem. Suppose a datagram channel $\mathcal{C}$ satisfies INT-CTXT. We play the FEP-CCA game, which we refer to as $E_0$. Suppose the probability that an adversary wins the INT-CTXT game is bounded by $p_{\text{INT}}$. Then we transition to a new hybrid experiment $E_1$ where $\mathcal{O}^b_{\tf{\fepdgramrecv}}$ simply returns $\bot$ to all queries. Consider in $E_0$, an adversary $A$ that produces some value $c \neq \bot$ from $\mathcal{O}^b_{\tf{\fepdgramrecv}}$. This is only possible if $b=0$, and in this case, $A$ if run as an adversary in the INT-CTXT game wins that game as well; in the $b=0$ case, the \tf{Send} oracles in the two games are identical. Thus, the transition from $E_0$ to $E_1$ is bounded by $p_{\text{INT}}$, which is negligible. Then, since in $E_1$ $\mathcal{O}^b_{\tf{\fepdgramrecv}}$ deterministically returns $\bot$, it easy to see that an adversary for $E_1$ can be run in $\mathsf{Exp}^{\textsf{\fepdgramsend},b}_{\mathcal{A}}(1^{\lambda})$ to win with with equal probability, completing the reduction. 
\end{proof}

We prove Theorem~\ref{prop:datagram-fepcca-indcca}
\begin{proof}
We observe that FEP-CCA trivially implies FEP-CPA, and implies INT-CTXT by Theorem~\ref{prop:datagram-cca-integrity}. Then we apply the argument from Theorem \ref{prop:lenreg-fep-ind} which gives that FEP-CPA and datagram channel length regularity imply IND-CPA, noting that datagram channel length regularity gives that any two equal-length inputs to the left-or-right oracle in the IND-CPA game have the same output length, just as in the datastream case.  Finally, since our definitions in the datagram setting are the traditional atomic security definitions, we apply the generic composition theorem from Bellare and Namprempre~\cite{bellare2008authenticated} which gives that INT-CTXT and IND-CPA imply IND-CCA. 
\end{proof}

\section{Proofs for Datagram Constructions}
\label{appendix:datagram-construction-proofs}
In this section we prove Theorem \ref{thm:datagram-construction-all} by establishing each individual property in the results below. 

\begin{theorem} \label{prop:datagram-construction-correctness} The channel construction in Figure \ref{fig:construction-datagram} satisfies correctness for Datagram channels.
\end{theorem}
\begin{proof}
First we consider Message Length. In the case where $p \geq 0$, the result directly follows from Theorem \ref{prop:datagram-shaping}. Otherwise, if $m \leq \inLength$ or $m = \nullsym$, \tf{Send} returns immediately on Line 3 or 5 with a valid output length.

For Message Delivery, we consider cases. First, if $|c| < \nullLength$, we observe \tf{Send} can only produce output by returning on Line 7, meaning $m = \nullsym$. An output message of this length will always be parsed by \tf{Recv} as $\nullsym$. Second, if $m = \nullsym$, we have a ciphertext, produced in Line 3 or 9, which begins with $0$. This is parsed as a null message in Line 6 of \tf{Recv}. Finally, if $m \neq \nullsym$, \tf{Send} produces a ciphertext in Line 5 or 13. Each ciphertext contains the $1$ byte, a two-byte length header, some optional padding bytes, and the message. \tf{Recv} parses this plaintext by extracting the length field, and extracting the plaintext bytes from the end of the message. 

For Message Acceptance, we consider first the cases with positive lengths. first suppose $m = \nullsym$ and let $p_{\nullsym}=0$; Line 7 returns the empty byte string, which is interpreted as $\nullsym$ by \tf{Recv}, which remains true for $0 \leq p < \nullLength$. Then if $\nullLength \leq p \leq \outLength$, \tf{Send} returns $c \neq \bot$ on Line 8. 

Second, let $m \leq \inLength$, $m \neq \nullsym$.  First we identify an output length $p_m = \aeadOverhead + 3 + |m|$ which will permit $m$ to be sent. $p_m \leq \outLength$ since $m \leq \inLength$. Then, we observe the function call $\tf{Send}(\tok[S]{st}, m ,p_m)$: none of the conditionals trigger, and the plaintext is produced on Line 13 and encrypted as the result, which is not $\bot$. 

Finally, if $p < 1$, the first two conditionals cover all messages in the message space, and each returns a non-error ciphertext.  
\end{proof}

\begin{theorem}\label{prop:datagram-construction-fepcpa}The channel construction in Figure \ref{fig:construction-datagram} satisfies FEP-CPA. \end{theorem}
\begin{proof}

The result follows from the fact that \tf{Send} in Figure \ref{fig:construction-datagram} always returns a direct call to \tf{Rand} or the output of our AEAD scheme, with a fresh nonce (generated at random when the \tf{Send} call is initiated). The probability of a nonce collision is negligible, and given no nonce collision, the IND-\$CPA property of our AEAD scheme guarantees that these outputs of \tf{Send} are computationally indistinguishable from random. 
\end{proof}

\begin{theorem}\label{prop:datagram-construction-chreg} The channel construction in Figure \ref{fig:construction-datagram} satisfies Datagram Channel Length Regularity \end{theorem}
\begin{proof}
We simply observe that our construction is stateless, and that for non-null messages, the output of the \tf{Send} function does not depend on the message content. For null messages, we note that in our definition either both messages are $\nullsym$, (and thus have the same output length in our construction given the same $p$) or neither are. 

\end{proof}

\begin{theorem} \label{prop:datagram-construction-int} The Datagram construction in Figure \ref{fig:construction-datagram} satisfies INT-CTXT. \end{theorem}

\begin{proof}
Since our protocol is stateless, we consider inputs to the \tf{Recv} oracle by message type. Each message falls into one of the following categories. First, it is smaller than $\nullLength$, in which case \tf{Recv} returns $\nullsym$ in Line 2, and the oracle converts this to $\bot$ which is given to the adversary. Second, it is produced by \tf{Send} and longer than $\nullLength$, in which case it is suppressed. Third, it was not produced by \tf{Send}, in which case \tf{Dec} fails except with negligible probability by the AUTH property of our AEAD scheme. 

Thus, every input to the \tf{Recv} oracle in the INT-CTXT game produces $\bot$ except with negligible probability. 
\end{proof}

\begin{theorem} \label{prop:datagram-shaping} The channel construction in Figure \ref{fig:construction-datagram} satisfies \shaping.
\end{theorem}

\begin{proof}
For the first requirement, suppose $c \neq \bot$ and $p \geq 0$. Thus, \tf{Send} returns on Line 7, 9, or 13. Each of these outputs has length $p$: on Line 7 the result is trivial, and for the remaining two outputs we simply observe that for our AEAD scheme, for any two messages, $|\tf{Enc}_k(\tok{nonce},m_1\Vert m_2)| = |\tf{Enc}_k(\tok{nonce},m_1)| + |m_2|$. 

For  the second, suppose $m = \nullsym$, $p \leq \outLength$. Then if $p<0$, the \tf{Send} produces a non-error result immediately on Line 3. Otherwise, the remaining ranges for $p$ are covered in the cases defined on Line 6 and 8, returning a non-error result in each case. 

\end{proof}

\begin{theorem} \label{thm:datagram-construction-fepcca}
The channel construction in Figure \ref{fig:construction-datagram} satisfies FEP-CCA and IND-CCA.
\end{theorem}
\begin{proof}
Theorems \ref{prop:datagram-construction-fepcpa},  \ref{prop:datagram-construction-chreg},  and \ref{prop:datagram-construction-int} establish the premises to apply Theorems \ref{prop:datagram-composition} and \ref{prop:datagram-fepcca-indcca}, giving the desired results. 
\end{proof}

\section{Detailed Descriptions of Analysis and Experiments on Real-World FEPs}
\label{appendix:experiments}

Here we fill in the details for the experiments we performed to determine the close behavior for real-world FEPs, and analysis of the specific message formats and reasoning we used to arrive at our results regarding minimum lengths. We first describe our general experimental design, and then give the details for each protocol. 

\subsection{Experimental Setup}
Unless otherwise stated, experiments were performed in a Pop! OS 22.04 Linux
environment running Linux Kernel 6.6.10, and messages were sent locally, via the
IPv4 loopback interface. We wrote custom software to run our experiments which
we make publicly available at~\url{https://github.com/efenske/bytes_to_schlep/}. In the datastream setting, each protocol consists of two applications, which we refer to as a Server Proxy and Client Proxy. We wrote a client, server, and an intermediate meddler-in-the-middle (MiTM) proxy to modify bytes in transit between the Client Proxy and Server Proxy. The MiTM proxy can be configured with the direction of the message to modify (towards the server, which we referred to as `out', `in' otherwise), the message to modify (specifying the ordinal $n$ to modify the $n$th message), and the byte within the message to change. By default, we modified the target byte by replacing it with its XOR by the value 1. We introduce two additional transparent MiTM proxies between the server and client themselves and their proxy applications  which relayed all connections and messages unchanged, to facilitate easy analysis of packet captures (thus, every connection in the experiment is connected to at least one static, bound port, making it trivial to later determine the sender and receiver of any message or fragment by its source and destination ports alone). Our experiments initialized all MiTM and proxy servers in reverse order (beginning with Server and ending with the MiTM between the Client and Client Proxy). Then we ran the client application, which sent a sequence of ten messages, each with a user-determined quantity of a distinct plaintext byte. The server responded to each client message by sending a message with the same repeated byte it received in a different quantity. By default to the client sent 600 bytes in each message, and the server sent 500. We chose these values since they were able to be encapsulated in a single message or ciphertext fragment for all of our analyzed protocols, were unlikely to be fragmented or merged by the operating system due to their size, and we chose the client and server length to be distinct to make packet captures easy to read. During experiments to validate minimum lengths, the client and server both sent instead messages of length 1. We ran TShark during all experiments listening on the loopback IPv4 interface writing to a file. The packet captures included the local timestamp for each frame, source and destination ports, message lengths, and the TCP message type (with specific fields for TCP FIN and TCP RST). After the experiment concluded, we wrote software to analyze the packet capture and experiment log, extract the packets within the valid time window for each individual experiment, and identify the sender and receiver of the first TCP FIN or RST frame by the destination and source port. Since many protocols include some handshake or initial state which transitions later to a data transport phase, we ran our experiments both modifying the initial few messages (n=1,2,3) and also messages later in the stream (n=5,10). Once we identified a close behavior, we ran the experiment multiple times, modifying different bytes in the fragment to confirm the result before indicating the close behavior on Table~\ref{tab:feps}. 

In the datagram setting and for kcptun, our experimental design varied, since these protocols required slightly different deployment environments. 

To assist the reader in following our reasoning and in reproducing our source code analysis, we often reference specific lines of source code. We give them by relative path from the root of the project repository (which is linked as a reference for each protocol), and we refer to the line numbers present in the version of repository as of February 6, 2024. 

\subsection{Protocols}

\textbf{Shadowsocks-libev}. Shadowsocks-libev~\cite{shadowsockslibev} is an implementation of Shadowsocks~\cite{shadowsocks} available in the official Debian repository. It is written in C and designed to be lightweight. Shadowsocks is an encrypted SOCKS5 proxy that requires a pre-shared key. The client begins the protocol by transmitting an IV, after which messages are AEAD encrypted, with each ciphertext block prepended by a two-byte length block, which is itself also AEAD encrypted. The Shadowsocks-libev source code explicitly includes logic in the design of the server to handle close behavior by way of a special state called STAGE\_STOP (see \eg Line 971 of src/server.c) which is entered after receiving a message that fails to decrypt, during which all messages are ignored and the connection is held open until terminated by the client. We deployed Shadowsocks-libev version 3.3.5 using the ChaChaPoly1305 cipher suite as recommended in the documentation.  For channel closures, we observed that in the outgoing direction, the server kept the connection open without terminating the connection throughout the length of the experiment, but in the incoming direction (\ie messages directed at the client), the client immediately terminates the TCP connection upon receiving a modified message (whether the modification occurs in the payload or tag of either block in a given message). Shadowsocks does produce error messages when messages fail to authenticate, and so its behavior in the outgoing direction does not technically satisfy FEP-CCFA, but with minor protocol modifications or an application that simply ignores errors, the definition is plausibly satisfied. Shadowsocks does not pad its messages. Since the authentication tag for this stream cipher has length 16, we expected the minimum length of an output fragment under this protocol to be 35 (two tags, two bytes of length data, and a one byte payload) and indeed observed this behavior in our experiments.

Since the protocol applies authenticated encryption to all messages using a cipher scheme which satisfies IND\$-CPA, it plausibly could satisfy FEP-CCFA, but the client's close behavior is incompatible with any secure close function. 

\textbf{V2Ray Framework}. The V2Ray framework~\cite{v2ray} implements a variety of protocols for the purposes of censorship circumvention including two FEPs: VMess~\cite{vmess}, and Shadowsocks. Shadowsocks is an alternate implementation of the same protocol described above, while VMess is a complex custom FEP which includes many anti-censorship features. It consists of an early phase, where initial and header messages are sent in both directions, followed by a data transport phase. We assessed both protocols as implemented in the V2Ray framework version 5.7.0. and configured both protocols using the sample configuration files. The V2Ray framework implements a so-called ``drain'' approach to frustrate active probing attempts, wherein a random number of bytes are selected for each connection, randomized on a per-user and per-connection basis, and the connection is not dropped until after this number of bytes have been read throughout the duration of the connection. In the course of our experiments on these V2Ray protocols we identified behavior that violates traditional security notions for encrypted protocols: they do not satisfy integrity. We observed that during the data transport phase of each protocol, if a byte within the ciphertext containing data is modified so that the authentication tag fails to validate, the payload is dropped but the connection is not closed, meaning that both protocols continue to parse new messages, so an active network adversary can by modifying appropriate bytes in transit drop certain messages from the stream. We disclosed our results and in particular this vulnerability to the framework maintainers on January 4, 2024 including a minimal working example of the attack. Thus, while both protocols satisfy FEP-CPFA, neither can satisfy FEP-CCFA since they do not satisfy INT-CST. 

Specific to Shadowsocks, we observed that the Shadowsocks protocol as implemented does not apply the draining behavior on the client side, and instead terminates the connection immediately if it receives an error. This has been confirmed to us as an known bug by the project maintainer. On the server side, the protocol drains as expected in its initial state, and in the data transport phase terminates the connection immediately when a length block fails to decrypt and authenticate. The V2Ray implementation of Shadowsocks uses the same message format as that of Shadowsocks-libev, and so produces the same 35 byte minimum length messages, which we observed in our experiments. 

Specific to VMess, we note that VMess includes a random string of random padding at the end of transmitted messages which is not authenticated, and so may be modified in transit without altering any protocol behavior. Finally, during its data transport phase, VMess prepends each encrypted ciphertext with a two byte encrypted but unauthenticated length field, meaning the length can be modified by an active adversary. Because of the integrity issues identified above, in practice modifying a length block will cause the protocol to hang indefinitely, since all following ciphertexts will fail to authenticate and be dropped silently. The protocol maintainer informed us that VMess sends keepalive messages during its data transport phase, which consist of a two byte length header along with a 16 byte authentication tag and an empty payload for a minimum length of 18 bytes, along some number of additional padding bytes, which is selected at random from 0-63 (see Line 114 of proxy/vmess/encoding/auth.go), making the minimum size 18 which we observed in our experiments.

\textbf{(lib)InterMAC}. (lib)InterMAC is a protocol proposed by Boldyreva \etal~\cite{boldyreva2012security} and implemented by Albrecht \etal~\cite{albrecht2019libintermac} which modifies the SSH Binary Packet Protocol in order to provide security properties associated with ciphertext lengths and fragmentation. The protocol transmits encrypted and authenticated fixed-sized chunks with the chunk size predetermined as a configuration parameter, padding to achieve this length if necessary. In response to our queries, the maintainer updated the libInterMAC repository with a complete implementation of the protocol on September 27, 2023, which is the version we tested. When a chunk fails to authenticate, libInterMAC terminates the connection in either direction as we verified in our experiments. We note that this protocol is not in use, and only encrypts the data transport phase of an SSH connection meaning that the existing InterMAC implementation is not properly a FEP, but we considered it valuable to include it here as an example of a distinct and reasonable approach to this problem from the others discussed here.

\textbf{Obfs4}. Obfs4~\cite{obfs4} is a fully encrypted transport protocol for Tor~\cite{dingledine2004tor}. It is a complex protocol that includes a key exchange and resists active probing attacks. The data transport phase of the protocol includes AEAD encrypted ciphertext blocks prepended by a two byte unauthenticated length field which is encrypted by a stream cipher. We installed Obfs4Proxy version 0.1.3 from the official Debian repository and used ptadapter~\cite{ptadapter} to run the transport as a standalone proxy. When a ciphertext fails to decrypt, the connection is immediately terminated by the receiver. This also occurs if an adversary modifies the length field in transit and the resulting expected ciphertext of an incorrect length fails to decrypt, which we were able to verify in our experiments. Thus while the protocol plausibly satisfies INT-CST (if a length field is modified, the protocol produces only channel closures/errors) it does not satisfy FEP-CCFA, since this close function is not secure (it requires knowledge of the encrypted lengths).

Obfs4 includes random padding, and selects a random target output length. This random target is understood to be a value modulo the maximum frame length for the protocol (a fixed parameter), and is selected from a set of output lengths according to a weighted distribution, where the weights and length values are selected based on a random seed specific to each Obfs4 server, constructed when the server is first initialized and saved for future use. If the target output length is too small to add any padding, the protocol instead pads to a full frame (~1500 bytes) before adding the desired padding amount (see the padBurst function: transports/obfs4/obfs4.go Line 607). Obfs4 messages include a 2 byte packet length, a 2 byte frame length, a 1 byte packet type token, and a 16 byte authentication tag, with a minimum payload of 1  byte (whether padding or data). Since the minimum plaintext size is 1 and Obfs4 always includes at least one padding packet, the minimum output length is 44 (two packets with payload length 1; one with a byte of padding and one with data). However, each server has \emph{its own individual minimum length}, which is the smallest value larger than 44 in its list of possible target output lengths. We were able to confirm this value is produced by the protocol under these circumstances by producing a modified version of the Obfs4 binary which replaces the function call sampling a value from the server-specific distribution with the constant 44, and subsequently observed messages of this length in our experiments, while replacing this with 43 produces only messages padded to at least a full frame. 

\textbf{OpenVPN-XOR}. OpenVPN~\cite{openvpn} is a VPN protocol that can be patched to encrypt all outputs by XOR with the repetition of a pre-shared key~\cite{xor}. We analyzed OpenVPN version 2.6.8 as provided from the official Debian repository without the XOR patch applied (and assume that the XOR encryption should not modify message lengths) using the parameters set in the example configurations provided with the software, which deploy the proxy using the AES256-GCM cipher suite and with TLS mode turned on, with other configuration options mostly remaining as default. The obfuscation method in place is trivially identifiable and other work~\cite{xue2022openvpn} has carefully studied its close behavior, which is to terminate the connection after an authentication failure. We explain the structure of OpenVPN messages as described in the official documentation~\cite{openvpn-network} under these default settings. OpenVPN messages are either data channel messages or control channel messages. Control channel messages are a minimum of 40 bytes if SHA-1 is used for the protocol's HMAC (which is the default) and if the 4-byte timestamps intended to combat replay attacks are switched off (in our installation, they are on by default). This means adopting the default settings, control channel messages are 44 bytes. Version 2 data channel messages under these settings include a 26 byte header, including a two-byte length field. We did not observe Version 1 data channel messages in our experiments under the default configuration. OpenVPN typically encapsulates other protocols in the data channel, including their headers, but also directly sends messages of its own, including a fixed 16-byte ping message (explicitly defined in src/openvpn/ping.c Line 42), which was the smallest we observed in our experiments. Thus, the minimum length we identified in the TCP setting was 42. 

\textbf{Obfuscated OpenSSH}. Obfuscated OpenSSH~\cite{obfsSSH} is a patch to OpenSSH wherein key material is transmitted in the clear (or a PSK is used), after which the handshake protocol is fully encrypted (the binary packet protocol for SSH is fully encrypted without any modification). The approach of Obfuscated OpenSSH is used in the Psiphon~\cite{psiphon} anti-censorship application, designed for Windows and mobile devices. We analyzed Obfuscated OpenSSH as downloaded from the official repository, which is a modification of OpenSSH 5.2\_p1 but since the protocol includes complex handshake logic we ran it directly through a MiTM proxy to modify bytes in transit through the loopback interface and analyzed the resulting packet captures with Wireshark. OpenSSH terminates the connection immediately in both directions when a binary packet protocol message fails to authenticate, and we found in our experiments that this is no different after the obfuscation patch is applied. OpenSSH does not pad message lengths though this is permitted in the SSH specification. According to the RFC for the binary packet protocol, "The minimum size of a packet is 16 (or the cipher block size, whichever is larger) bytes (plus 'mac')"~\cite{ssh-rfc}. For the initial messages, before the session key is negotiated, the MAC field is empty, meaning we expected the minimum to be 16 bytes. We ran the protocol in obfuscated and non-obfuscated configuration, and identified 16 byte messages appearing as part of the initial handshake in both circumstances. In the non-obfuscated captures, we identified this as a New Keys message, which indicates the switch to the recently negotiated session keys. 

\textbf{kcptun}. kcptun~\cite{kcptun} is a custom datastream transport protocol available in the official Debian repository, implemented via UDP datagrams and designed for very noisy networks. The protocol uses error correcting codes and sends many redundant messages to achieve a high probability that the sequence of messages can be reconstructed on the receiver side even in degraded network conditions. kcptun packets are by default encrypted using AES in Cipher Feedback mode, with a Cyclic Redundancy Check applied to the plaintext and prepended before encryption, providing non-cryptographic authentication, meaning that while the messages are fully encrypted, the protocol does not satisfy FEP-CCFA since it does not satisfy INT-CST. When the CRC is invalid, the packet is dropped, and kcptun measures the number of dropped packets and adjusts its internal parameters accordingly, significantly and predictably increasing the rate at which redundant messages are sent in direct proportion to the network packet loss. We wrote distinct software to proxy UDP messages, and did not notice a change in protocol behavior with increased induced packet loss beyond a predictable change in volume of packet transmissions. kcptun does not include padding and  uses empty messages as keepalives which are the smallest that arose in our experiments, which are 16 (nonce) + 4 (CRC) + 24 (KCP headers) + 8 (Error Correction metadata)  bytes for a total of 52 byte keepalive messages. The message format for kcptun messages is complex: a more thorough description of the headers, encapsulation format, and final message layout is given in the project page for the dependency kcp-go~\cite{kcp-go}.

\textbf{Shadowsocks-libev (UDP)}. We also analyzed Shadowsocks-libev in its UDP configuration using the same version, cipher suite and default settings as in the datastream setting. The protocol simply drops packets that fail to authenticate. The Shadowsocks datagram message format includes a fresh 32-byte nonce at the beginning of each datagram, and AEAD encrypts the payload with a 16 byte authentication tag. In order for a Shadowsocks datagram message to be constructed the sender must in addition supply a target address in the SOCKS5 format, which consists of a one byte address type (IPv4, IPv6, or a string), the address, and a two byte port. Addresses are 4 bytes for an IPv4 address, producing a minimum length of an empty datagram as 32 + 16 + 7 = 55 bytes, but we note that if the address is given as a string, the length of the ciphertext depends on domain name of the address and thus the proxy can produce messages of length smaller than 55. Using one byte for the string length and an empty string as the destination domain, we were able to induce the proxy to produce messages of length 52, though we note this is only possible under pathological conditions, i.e. the server is asked to transmit empty messages to an invalid domain. Thus we consider 55 the minimum under standard conditions.

\textbf{SWGP-go}. SWGP~\cite{swgp} is a proxy designed for the Wireguard VPN protocol, intended to transform Wireguard into a datagram FEP. It behaves in two modes: one called ``zero overhead'', where each Wireguard packet has its first 16 bytes encrypted via AES-ECB using a pre-shared key, and ``paranoid'', which pads and AEAD encrypts each message. We assess the protocol in its paranoid mode. We deployed SWGP Version 1.5.0 configured with the default settings suggested by the maintainers on the project landing page. Since we suspected that running both multiple Wireguard and SWGP instances locally would be cumbersome, we ran our experiments through Docker-Compose~\cite{docker-compose} and constructed three Docker images, a client image (including the client Wireguard instance and the client SWGP application), a router image which passed on messages transparently while recording them via TShark, and a server image. The client image also included a Python program to send an empty, raw IP packet through the interface every 5 seconds. Each experiment ran for 5 minutes and we assessed the resulting packet captures recorded by the router image. We observed in the source code that the paranoid protocol padded each message to a random length between one more than the packet's minimum length and the MTU, which is set as a configuration option and is 1500 as default. Thus we hypothesized the minimum lengths to be the smallest messages in Wireguard, which are 32 byte keepalives, encapsulated with a 24 byte nonce and a 16 byte authentication tag, and a 2 byte length field (the encapsulation format is given in packet/paranoid.go, Line 18) with one required minimum padding byte for a total of 75, which is the smallest we observed in our experiments.

\textbf{OpenVPN-XOR (UDP)}. This protocol is nearly identical to the datastream version of OpenVPN, and has been similarly analyzed in literature~\cite{xue2022openvpn}. Again, we identify the minimum lengths using the same approach, though in the UDP setting there is no two-byte length field prepended to every message, meaning that the length of the keepalives (and thus the minimum length) is 40, which we observed in our experiments. 

\fi
\end{document}